\documentclass[format=acmsmall, screen=true]{acmart}
\usepackage{subcaption}
\usepackage{xspace}
\usepackage{todonotes}
\usepackage{mathtools}
\usepackage{natbib}
\usepackage{enumitem}
\usepackage[utf8]{inputenc}
\numberwithin{equation}{section}

\acmJournal{JEA}
\received{September 2018}

\graphicspath{{./fig/}}

\newcommand\OPT{\ensuremath{\mathrm{OPT}}\xspace}
\newcommand\TOP{\ensuremath{\mathrm{TOP}}\xspace}

\newcommand\MIN{\ensuremath{\mathrm{MIN}}\xspace}
\newcommand\BOT{\ensuremath{\mathrm{BOT}}\xspace}

\newcommand\CMP{\ensuremath{\mathrm{CMP}}\xspace}
\newcommand\CMPS{\ensuremath{\mathrm{CMP}(\mathcal{Q}^*)}\xspace}

\newcommand\STP{\ensuremath{\mathrm{ST}}\xspace}

\DeclarePairedDelimiter{\floor}{\lfloor}{\rfloor}
\newcommand{\myparagraph}[1]{\medskip\noindent{\sffamily\bfseries #1}~}

\begin{document}

\title{Multi-Level Steiner Trees}


\author{Reyan~Ahmed}
\affiliation{%
  \institution{University of Arizona}
  \city{Tucson}
  \state{AZ}
  \country{USA}}
\email{abureyanahmed@email.arizona.edu}

\author{Patrizio~Angelini}
\orcid{0000-0002-7602-1524}
\affiliation{%
  \institution{Universit\"at T\"ubingen}
  \city{T\"ubingen}
  \country{Germany}}

\author{Faryad~Darabi~Sahneh}
\orcid{0000-0003-2521-4331}
\affiliation{%
  \institution{University of Arizona}
  \city{Tucson}
  \state{AZ}
  \country{USA}}
\email{faryad@cs.arizona.edu}

\author{Alon~Efrat}
\affiliation{%
  \institution{University of Arizona}
  \city{Tucson}
  \state{AZ}
  \country{USA}}
\email{alon@cs.arizona.edu}

\author{David~Glickenstein}
\affiliation{%
  \institution{University of Arizona}
  \city{Tucson}
  \state{AZ}
  \country{USA}}
\email{glickenstein@math.arizona.edu}

\author{Martin~Gronemann}
\affiliation{%
  \institution{Universit\"at zu K\"oln}
  \city{Cologne}
  \country{Germany}}
\email{gronemann@informatik.uni-koeln.de}

\author{Niklas~Heinsohn}
\affiliation{%
  \institution{Universit\"at T\"ubingen}
  \city{T\"ubingen}
  \country{Germany}}
\email{heinsohn@informatik.uni-tuebingen.de}

\author{Stephen~G.~Kobourov}
\orcid{0000-0002-0477-2724}
\affiliation{%
  \institution{University of Arizona}
  \city{Tucson}
  \state{AZ}
  \country{USA}}
\email{kobourov@cs.arizona.edu}

\author{Richard~Spence}
\affiliation{%
  \institution{University of Arizona}
  \city{Tucson}
  \state{AZ}
  \country{USA}}
\email{rcspence@email.arizona.edu}

\author{Joseph~Watkins}
\affiliation{%
  \institution{University of Arizona}
  \city{Tucson}
  \state{AZ}
  \country{USA}}
\email{jwatkins@math.arizona.edu}

\author{Alexander~Wolff}
\orcid{0000-0001-5872-718X}
\affiliation{%
  \institution{Universit\"at W\"urzburg}
  \city{W\"urzburg}
  \country{Germany}}



\begin{abstract}
  In the classical Steiner tree problem, given an
  undirected, connected graph $G=(V,E)$ with non-negative edge costs 
  and a set of \emph{terminals} $T\subseteq V$, the objective is to find a minimum-cost tree $E' \subseteq E$ that spans the terminals.  The problem is APX-hard;
  the best known approximation 
  algorithm has a ratio of $\rho = \ln(4)+\varepsilon < 1.39$.
  In this paper, we study a natural generalization, the \emph{multi-level
    Steiner tree} (MLST) problem:
  given a nested sequence of
  terminals $T_{\ell} \subset \dots \subset T_1 \subseteq V$,
  compute nested trees
  $E_{\ell}\subseteq \dots \subseteq E_1\subseteq E$ that span the
  corresponding terminal sets with minimum total cost.

  The MLST problem and variants thereof have been 
  studied under various names including Multi-level Network Design, 
  Quality-of-Service Multicast tree, 
  Grade-of-Service Steiner tree, and Multi-Tier tree.  
  Several approximation results are known.  
  We first present two simple $O(\ell)$-approximation heuristics. Based on these, we introduce a rudimentary composite 
  algorithm that generalizes the above heuristics, and determine its approximation
  ratio by solving a linear program. 
  We then present a method that guarantees the same approximation ratio 
  using at most $2\ell$ Steiner tree computations.
  We compare these heuristics experimentally on various instances of up to 500
  vertices using three different network generation models.
  We also present various integer linear programming (ILP) formulations for the MLST problem, and
  compare their running times on these instances.
  To our knowledge, the composite algorithm achieves the best approximation ratio for up to $\ell=100$ levels,
  which is sufficient for most applications such as network visualization or designing multi-level infrastructure.
\end{abstract}


\maketitle

\renewcommand{\shortauthors}{R.~Ahmed et al.}

\begin{acks}
  This work is partially supported by \grantsponsor{nsf}{NSF}{https://nsf.gov} 
  grants \grantnum{nsf}{CCF-1423411} and \grantnum{nsf}{CCF-1712119}.

  The authors wish to thank the organizers and participants of the
  First Heiligkreuztal Workshop on Graph and Network Visualization where work on this problem began.
\end{acks}

\section{Introduction}

Let $G = (V, E)$ be an undirected, connected graph with positive edge costs $c \colon E \to \mathbb{R}^+$, and let $T \subseteq V$ be a set of vertices called \emph{terminals}. A \emph{Steiner tree} is a tree in~$G$ that spans~$T$. 
The \emph{network (graph) Steiner tree problem} (\STP) is to find a minimum-cost Steiner tree $E' \subseteq E$, where the cost of~$E'$ is $c(E') = \sum_{e \in E'} c(e)$. 
\STP is one of Karp's initial NP-hard problems \cite{Karp1972}; see also a survey~\cite{Winter1987}, an online compendium~\cite{Hauptmann2015}, and a textbook~\citep{Proemel2002}.

Due to its practical importance in many domains, 
there is a long history of exact and approximation algorithms for the problem.
The classical 2-approximation algorithm for \STP \cite{Gilbert1968}
uses the \emph{metric closure} of~$G$, i.e., the complete edge-weighted graph~$G^*$ with vertex set~$T$ in which, for every edge $uv$, the cost of~$uv$ equals the length of a shortest $u$--$v$ path in~$G$. 
A minimum spanning tree of~$G^*$ corresponds to a 2-approximate Steiner tree in~$G$.

Currently, the last in a long list of improvements is the LP-based approximation algorithm of Byrka et al.~\cite{Byrka2013}, which has a ratio of $\ln(4)+\varepsilon < 1.39$.
Their algorithm uses a new iterative randomized rounding technique.
Note that \STP is APX-hard \cite{Bern1989};
more concretely, it is NP-hard to approximate the problem within a factor of $96/95$~\cite{Chlebnik2008}.
This is in contrast to the geometric variant of the problem, where terminals correspond to points in the Euclidean or rectilinear plane.
Both variants admit polynomial-time approximation schemes (PTAS) \cite{Arora1998,Mitchell1999}, while this is not true for the general metric case \cite{Bern1989}.

In this paper, we consider the natural generalization of \STP where the terminals appear on ``levels'' (or ``grades of service'') and must be connected by edges of appropriate levels.
We propose new approximation algorithms and compare them to existing ones both theoretically and experimentally.

\begin{definition}[Multi-Level Steiner Tree (MLST) Problem] 
Given a connected, undirected graph $G = (V,E)$ with edge weights $c \colon E \rightarrow \mathbb{R}^+$ and $\ell$ nested terminal sets $T_\ell \subset \dots \subset T_1 \subseteq V$, a \emph{multi-level Steiner tree} consists of $\ell$ nested edge sets $E_\ell \subseteq \cdots \subseteq E_1 \subseteq E$ such that $E_i$ spans $T_i$ for all $1 \le i \le \ell$. The cost of an MLST is defined by the sum of the edge weights across all levels, $\sum_{i=1}^{\ell} c(E_i) = \sum_{i=1}^{\ell} \sum_{e \in E_i} c(e)$. The MLST problem is to find an MLST $E_{\OPT,\ell} \subseteq \dots \subseteq E_{\OPT,1} \subseteq E$ with minimum cost.
\end{definition}

Since the edge sets are nested, the cost of an MLST equivalently equals $\sum_{e \in E} L(e)c(e)$, where $L(e)$ denotes the highest level that edge $e$ appears in, where $L(e)=0$ if $e \not\in E_1$. This emphasizes that the cost of each edge is multiplied by the number of levels it appears on.

We denote the cost of an optimal MLST by \OPT.  
We can write $$\OPT=\ell\OPT_\ell+(\ell-1)\OPT_{\ell-1}+\dots+\OPT_1$$ 
where $\OPT_\ell = c(E_{\OPT, \ell})$ and $\OPT_i=c(E_{\OPT,i}\backslash E_{\OPT,i+1})$ for $\ell-1 \ge i \ge 1$. Thus $\OPT_{i}$ represents the cost of edges on level $i$ but not on level $i+1$ in the minimum cost MLST.
%
%
Figure~\ref{fig:mlst_example} shows an example of an MLST for $\ell=3$.

\begin{figure}[tb]

{\color{gray}\frame{\includegraphics[width=.23\textwidth, clip, trim=0cm 13.75cm 0cm 0cm]{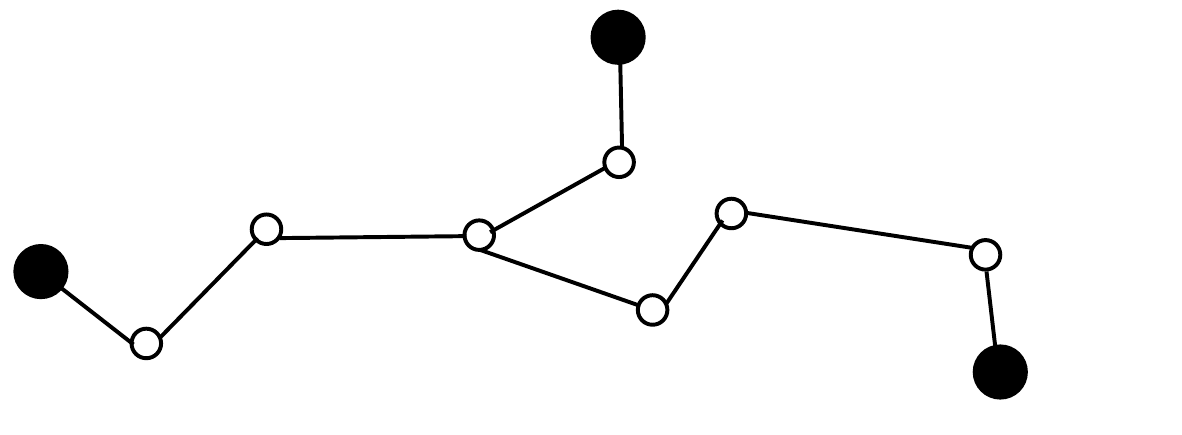}}
\hfill  \frame{\includegraphics[width=.23\textwidth, clip, trim=0cm 9.25cm 0cm 4.5cm]{mlst_example}}
\hfill
\frame{\includegraphics[width=.23\textwidth, clip, trim=0cm 4.8cm 0cm 8.95cm]{mlst_example}}
\hfill \frame{\includegraphics[width=.23\textwidth, clip, trim=0cm 0cm 0cm 13.75cm]{mlst_example}}}

  \caption{An illustration of an MLST with $\ell=3$ for the graph at the right. Solid and open circles represent terminal and non-terminal nodes, respectively. Note that the level-3 tree (left) is contained in the level-2 tree (mid), which is in turn contained in the level-1 tree (right).}
  \label{fig:mlst_example}
\end{figure}

\myparagraph{Applications.}
This problem has natural 
applications in designing multi-level infrastructure of low cost. 
Apart from this application in network design, multi-scale representations of graphs are useful in applications such as 
network visualization, where the goal is to represent a given graph 
at different levels of detail.  

\myparagraph{Previous Work.}
Variants of the MLST problem have been studied previously under various names, such as \emph{Multi-Level Network Design (MLND)}~\cite{Balakrishnan1994}, \emph{Multi-Tier Tree (MTT)}~\cite{mirchandani1996MTT}, \emph{Quality-of-Service (QoS) Multicast Tree}~\cite{1288137}, and \emph{Priority-Steiner Tree}~\cite{Chuzhoy2008}.

In MLND, the vertices of the given graph are partitioned into $\ell$ levels, and the task is to construct an $\ell$-level network.
Each edge $(i,j) \in E$ can contain one of $\ell$ different facility types (levels), each with a different cost (denoted ``secondary'' and ``primary'' with costs $0 \le b_{ij} \le a_{ij}$ for 2 levels).
The vertices on each level must be connected by edges of the corresponding level or higher, and edges of higher level are more costly.
The cost of an edge partition is the sum of all edge costs, and the task is to find a partition of minimum cost. Let~$\rho$ be the ratio of the best approximation algorithm for (single-level) \STP,
that is, currently $\rho=\ln(4)+\varepsilon < 1.39$. Balakrishnan et al.~\cite{Balakrishnan1994} gave a $(4/3)\rho$-approximation algorithm for 2-level MLND with proportional edge costs.
Note that the definitions of MLND and MLST treat the bottom level differently.
While MLND requires that \emph{all} vertices are connected eventually, this is not the case for MLST.

For MTT, which is equivalent to MLND, Mirchandani~\cite{mirchandani1996MTT}
presented a recursive algorithm that involves $2^\ell$ Steiner tree computations. 
For $\ell=3$, the algorithm achieves an approximation ratio of $1.522\rho$ independently of the edge costs $c^1,\dots,c^\ell \colon E \rightarrow \mathbb{R}^+$.
For proportional edge costs, Mirchandani's analysis yields even an approximation ratio of $1.5\rho$ for $\ell=3$.
Recall, however, that this assumes $T_1=V$, and setting the edge costs on the bottom level to zero means that edge costs are \emph{not} proportional.

In the QoS Multicast Tree problem~\cite{1288137} one is given a graph, a source vertex~$s$, and a level between~1 and~$k$ for each terminal (1~for highest priority).
The task is to find a minimum-cost Steiner tree that connects all terminals to~$s$.
The level of an edge~$e$ in this tree is the minimum over the levels of the terminals that are connected to~$s$ via~$e$.
The cost of the edges and of the tree are as above.
As a special case, Charikar et al.~\cite{1288137} studied the \emph{rate model}, where edge costs are proportional, and show that the problem remains NP-hard if all vertices (except the source) are terminals (at some level).
Note that if we choose as source any vertex at the top level~$T_\ell$, then MLST  can be seen as an instance of the rate model. 

Charikar et al.~\cite{1288137} gave a simple $4\rho$-approximation algorithm for the rate model.
Given an instance~$\varphi$, their algorithm constructs an instance~$\varphi'$ where the levels of all vertices are rounded up to the nearest power of~2. 
Then the algorithm simply computes a Steiner tree at each level of~$\varphi'$ and prunes the union of these Steiner trees into a single tree.
The ratio can be improved to $e\rho$, where $e$ is the base of the natural logarithm, using randomized doubling. 

Instead of taking the union of the Steiner trees on each rounded level, Karpinski et al.~\cite{Karpinski2005} contract them into the source in each step, which 
 yields a $2.454\rho$-approximation. 
They also gave a $(1.265+\varepsilon)\rho$-approximation for the 2-level case. 
(Since these results 
are not stated with respect
to~$\rho$, but depend on several Steiner tree approximation
algorithms~-- among them the best approximation algorithm with ratio
1.549~\cite{Robins2005} available at the time~-- we obtained the
numbers given here by dividing their results by 1.549 and stating the
factor~$\rho$.) 

For the more general Priority-Steiner Tree problem, where edge costs
are not necessarily proportional, Charikar et al.~\cite{1288137} gave
a $\min\{2 \ln |T|, \ell\rho\}$-approximation algorithm.
Chuzhoy et al.~\cite{Chuzhoy2008} showed that Priority-Steiner Tree does not admit an $O(\log\log n)$-approximation algorithm unless NP$\,\subseteq\,$DTIME$(n^{O(\log\log\log n)})$. 
For Euclidean MLST,
Xue at al.~\cite{Xue2001} gave a recursive algorithm that uses any algorithm for Euclidean Steiner Tree (EST) as a subroutine.  
With a PTAS \cite{Arora1998,Mitchell1999} for EST, the approximation ratio of their algorithm is $4/3+\varepsilon$ for $\ell=2$ and $(5+4\sqrt{2})/7+\varepsilon \approx 1.522+\varepsilon$ for $\ell=3$.

\myparagraph{Our Contribution.}
We give two simple approximation algorithms for
MLST, bottom-up and top-down, in
Section~\ref{subsection:heuristics}.  The bottom-up heuristic uses a
Steiner tree at the bottom level for the higher levels after pruning
unnecessary edges at each level.  The top-down heuristic first
computes a Steiner tree on the top level.  Then it passes edges down
from level to level until the bottom level terminals are spanned.

In Section ~\ref{subsection:composite}, we propose a composite heuristic that generalizes these, by 
examining all possible $2^{\ell-1}$ (partial) top-down and bottom-up
combinations and returning the one with the lowest cost. We propose a linear program that
finds the approximation ratio of the composite heuristic for any fixed
value of~$\ell$, and compute approximation ratios for up to
$\ell=100$ levels, which turn out to be better than those of previously known
algorithms.  However, the composite heuristic requires roughly $2^\ell \ell$ \STP
computations.  

Therefore, we propose a procedure that achieves the same approximation
ratio as the composite heuristic but needs at most $2\ell$ \STP computations.
In particular, it achieves a ratio of $1.5\rho$ for $\ell=3$ levels,
which settles a question posed by Karpinski et
al.~\cite{Karpinski2005} who were asking whether the
$(1.522+\varepsilon)$-approximation of Xue at al.~\cite{Xue2001} can be
improved for $\ell=3$.  Note that Xue et al.\ treated the Euclidean case,
so their ratio does not include the factor~$\rho$.
We generalize an integer linear programming (ILP) formulation for \STP
\cite{Polzin2001} to obtain an ILP formulation for the MLST problem in
Section~\ref{section:exact_algorithm}.
We experimentally evaluate several approximation and exact algorithms
on a wide range of problem instances in
Section~\ref{section:experimental_results}.  The results show that the
new algorithms are also surprisingly good in practice.  We conclude in
Section~\ref{section:conclusions}.

\section{Approximation Algorithms}
\label{section:approximation_algorithms}

In this section we propose several approximation algorithms for the MLST problem. In Section \ref{subsection:heuristics}, we show that the natural approach of computing edge sets either from top to bottom or vice versa, already give $O(\ell)$-approximations; we call these two approaches \emph{top-down} and \emph{bottom-up}, and denote their cost by \TOP and \BOT, respectively. Then, we show that running the two approaches and selecting the solution with minimum cost produces a better approximation ratio than either top-down or bottom-up.

In Section \ref{subsection:composite}, we propose a composite approach that mixes the top-down and bottom-up approaches by solving \STP on a certain subset of levels, then propagating the chosen edges to higher and lower levels in a way similar to the previous approaches. We then run the algorithm for each of the $2^{\ell-1}$ possible subsets, and select the solution with minimum cost. For all practically relevant values of $\ell$ ($\ell \leq 100$), our results improve over the state of the art.

\subsection{Top-Down and Bottom-Up Approaches}
\label{subsection:heuristics}

We present top-down and bottom-up approaches for computing approximate multi-level Steiner trees. The approaches are similar to the MST and Forward Steiner Tree (FST) heuristics by Balakrishnan et al.~\cite{Balakrishnan1994}; however, we generalize the analysis to an arbitrary number of levels.

In the top-down approach, an exact or approximate Steiner tree~$E_{\TOP,\ell}$ spanning the top level $T_\ell$ is computed. Then we modify the edge weights by setting $c(e) := 0$ for every edge $e \in E_{\TOP,\ell}$.  In the resulting graph, we compute a Steiner tree~$E_{\TOP,\ell-1}$ spanning the terminals in~$T_\ell-1$. This extends $E_{\TOP,\ell}$ in a greedy way to span the terminals in $T_\ell-1$ not already spanned by $E_{\TOP,\ell}$. Iterating this procedure for all levels yields a solution $E_{\TOP,\ell} \subseteq \cdots \subseteq E_{\TOP,1} \subseteq E$ with cost $\TOP$.

In the bottom-up approach, a Steiner tree~$E_{\BOT,1}$ spanning the bottom level $T_1$ is computed. This induces a valid solution for all levels. We can ``prune'' edges by letting $E_{\BOT,i}$ be the smallest subtree of $E_{\BOT,1}$ that spans all the terminals in $T_i$, giving a solution with cost $\BOT$. Note that the top-down and bottom-up approaches involve $\ell$ and 1 Steiner tree computations, respectively.

A natural approach is to run both top-down and bottom-up approaches and select the solution with minimum cost. This yields an approximation ratio better than those from top-down or bottom-up. Let $\rho \ge 1$ denote the approximation ratio for \STP (that is, $\rho = 1$ corresponds to using an exact \STP subroutine).
Let $\MIN_i$ denote the cost of a minimum Steiner tree over the terminal set 
$T_{i}$ with original edge weights, independently of other levels, so that $\MIN_{\ell} \le \MIN_{\ell-1} \le \ldots \le \MIN_1$. Then $\OPT \ge \sum_{i=1}^{\ell} \MIN_i$ trivially.

\begin{theorem} \label{theorem:heuristics}
For $\ell \ge 2$ levels, the top-down approach is an $\frac{\ell+1}{2}\rho$-approximation to the MLST problem, the bottom-up approach is an $\ell\rho$-approximation, and the algorithm returning the minimum of $\TOP$ and $\BOT$ is an $\frac{\ell+2}{3}\rho$-approximation.
\end{theorem}

In the following we give the  proof of Theorem~\ref{theorem:heuristics}. Let \TOP be the total cost produced by
the top-down approach, and let $\TOP_{i} = c(E_{\TOP, i} \backslash E_{\TOP, i+1})$ be the cost of edges on level $i$ but not on level $i+1$. Then $\TOP = \sum_{i=1}^{\ell} i \TOP_i$. Define $\BOT$ and $\BOT_{i}$ analogously.

\begin{lemma} \label{lemma:2level}
The following inequalities relate $\TOP$ with $\OPT$:
\begin{align}
  \TOP_\ell & \le \rho \OPT_\ell \label{eqn:2level-1}\\
  \TOP_{\ell-i} & \le \rho(\OPT_{\ell-i} + \ldots + \OPT_\ell) \text{ for all } 1 \le i \le \ell-1 \label{eqn:2level-2}
\end{align}
\end{lemma}
\begin{proof}
Inequality (\ref{eqn:2level-1}) follows from the fact that $E_{\TOP, \ell}$ is a $\rho$-approximation for \STP over $T_\ell$, that is, $\TOP_\ell \le \rho\MIN_\ell \le \rho\OPT_\ell$. To show (\ref{eqn:2level-2}), note that $\TOP_{\ell-i}$ represents the cost of the Steiner tree over terminals $T_{\ell-i}$ with some edges (those already included in $E_{\ell-i+1}$) having weight $c(e)$ set to zero. Then $\TOP_{\ell-i} \le \rho \MIN_{\ell-i}$. Since $E_{\OPT, \ell-i}$ spans $T_{\ell-i}$ by definition, we have $\MIN_{\ell-i} \le c(E_{\OPT, \ell-i}) = \OPT_{\ell-i} + \ldots + \OPT_{\ell}$. By transitivity, $\TOP_{\ell-i} \le \rho(\OPT_{\ell-i} + \ldots + \OPT_{\ell})$ as desired.
\end{proof}

Using Lemma \ref{lemma:2level}, we have an upper bound on $\TOP$ in terms of $\OPT_1, \ldots, \OPT_{\ell}$:
\begin{align*}
\TOP &= \ell \TOP_\ell + (\ell-1) \TOP_{\ell-1} + \ldots + \TOP_1 \\
&\le \ell\rho \OPT_{\ell} + (\ell-1)\rho(\OPT_{\ell-1} + \OPT_\ell) + \ldots + \rho(\OPT_1 + \OPT_2 + \ldots + \OPT_\ell) \\
&= \rho\left(\frac{(\ell+1)\ell}{2}\OPT_{\ell} + \frac{\ell(\ell-1)}{2} \OPT_{\ell-1} + \ldots + \frac{2 \cdot 1}{2} \OPT_1 \right) \\
&\le\frac{\ell+1}{2}\rho \cdot \OPT.
\end{align*}
Therefore the top-down approach is an $\frac{\ell+1}{2}\rho$-approximation. In Fig.~\ref{fig:examples-top} we provide an example showing that our analysis is tight for $\rho = 1$.

The bottom-up approach is a fairly trivial $\ell \rho$-approximation to the MLST problem, even without pruning edges. Consequentially, $\BOT \le \ell \cdot c(E_{\BOT,1})$ as pruning no edges results in a solution with cost $\ell \cdot c(E_{\BOT,1})$. 

As $E_{\BOT,1}$ is found by computing a Steiner tree over the bottom level $T_1$, we have $c(E_{\BOT,1}) \le \rho \MIN_1$. Additionally, $\MIN_1 \le c(E_{\OPT,1}) = \OPT_1 + \OPT_2 + \ldots + \OPT_{\ell}$ as $E_{\OPT,1}$ is necessarily a Steiner tree spanning $T_1$. Combining these inequalities yields
\begin{align*}
\BOT &\le \ell \cdot c(E_{\BOT,1}) \\
&\le \ell \rho \MIN_1 \\
&\le \ell \rho (\OPT_1 + \OPT_2 + \ldots + \OPT_{\ell}) \\
&\le \ell \rho \sum_{i=1}^{\ell} i \OPT_i \\
&= \ell \rho \cdot \OPT
\end{align*}

Again, the approximation ratio (for $\rho = 1$) is asymptotically tight; see Figure \ref{fig:examples-bot}.

We show that taking the better of the two solutions returned by
the top-down and the bottom-up approach provides a
$\frac{4}{3}\rho$-approximation to MLST for $\ell = 2$. To prove this,
we use the simple fact that $\min \{x,y\} \le \alpha x + (1-\alpha)y$ for all $x,y \in \mathbb{R}$ and $\alpha \in [0,1]$.  Using the previous results on the upper bounds for $\TOP$ and $\BOT$ for $\ell=2$:
\begin{align*}
\min \{\TOP,\BOT\} & \le \alpha (3\rho\,\OPT_2+\rho\,\OPT_1) +
  (1-\alpha)(2\rho \,\OPT_2 + 2\rho \,\OPT_1) \\
  & = (2+\alpha)\rho\, \OPT_2 + (2-\alpha)\rho\, \OPT_1
\end{align*}
Setting $\alpha = \frac{2}{3}$ gives $\min\{\TOP, \BOT\} \le \frac{8}{3}\rho\,\OPT_2 + \frac{4}{3}\rho\,\OPT_1 = \frac{4}{3}\rho\,\OPT$.

For $\ell > 2$ levels, using the same idea gives
\begin{align*}
    \min\{\TOP, \BOT\} &\le \alpha \rho \sum_{i=1}^{\ell} \frac{i(i+1)}{2} \OPT_i + (1-\alpha) \ell \rho \sum_{i=1}^{\ell} \OPT_i \\
    &= \sum_{i=1}^{\ell} \left[ \left(\frac{i(i+1)}{2} - \ell\right)\alpha + \ell \right] \rho \OPT_i \\
\end{align*}
Since we are comparing $\min\{\TOP, \BOT\}$ to $t \cdot \OPT$ for some approximation ratio $t > 1$, we can compare coefficients and find the smallest $t \ge 1$ such that the system of inequalities
	\begin{align*}
	\left(\frac{\ell(\ell+1)}{2} - \ell\right) \rho \alpha + \ell \rho &\le \ell t \\
	\left(\frac{(\ell-1)\ell}{2} - \ell\right) \rho \alpha + \ell \rho &\le (\ell-1)t \\[-0.8em]
	&\vdots \\[-0.8em]
	\left(\frac{2 \cdot 1}{2} - \ell\right) \rho \alpha + \ell \rho &\le t
	\end{align*}
	has a solution $\alpha \in [0,1]$. 
	Adding the first inequality to $\ell/2$ times the last inequality yields $\frac{\ell^2 + 2\ell}{2}\rho \le \frac{3\ell t}{2}$.  
	This leads to $t \ge \frac{\ell+2}{3}\rho$.  
	Also, it can be shown algebraically that $(t, \alpha) = (\frac{\ell+2}{3}\rho, \frac{2}{3})$ simultaneously satisfies the above inequalities. 
	This implies that $\min\{\TOP, \BOT\} \le \frac{\ell+2}{3}\rho \cdot \OPT$ and concludes the proof of Theorem~\ref{theorem:heuristics}.

\begin{figure}
    \centering
    \input{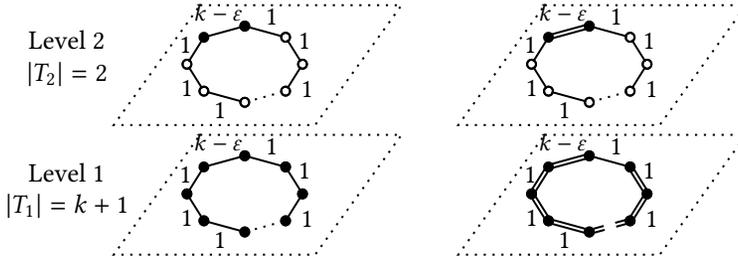}
    \caption{The approximation ratio of $\frac{\ell+1}{2}$ for the top-down approach is asymptotically tight. In the example above for $\ell=2$, the input graph (left) consists of a $(k+1)$-cycle with one edge of weight $k-\varepsilon$. The solution returned by top-down (right) has cost $\TOP = (k-\varepsilon) + (k-\varepsilon + k) \approx 3k$, whereas $\OPT = 2k$.}
    \label{fig:examples-top}
\end{figure}

\begin{figure}
    \centering
    \input{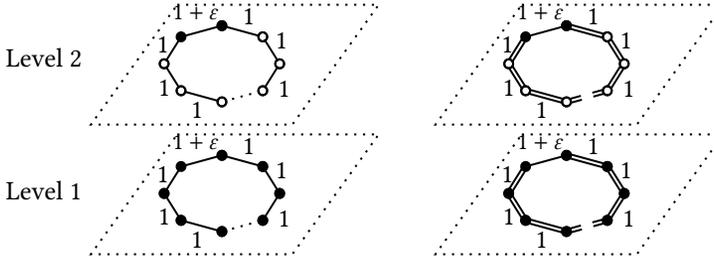}
    \caption{The approximation ratio of $\ell$ for the bottom-up approach is asymptotically tight. Using the same input graph as in Figure \ref{fig:examples-top}, except by modifying the edge of weight $k-\varepsilon$ so that its weight is $1+\varepsilon$, we see that $\BOT = 2k$, whereas $\OPT = (1+\varepsilon) + (1+\varepsilon + k-1) \approx k+1$.}
    \label{fig:examples-bot}
\end{figure}

Combining the graphs in Figures~\ref{fig:examples-top} and~\ref{fig:examples-bot} shows that our analysis of the combined top-down and bottom-up approaches (with ratio $\frac{4}{3}$) is asymptotically tight.

\subsection{Composite Algorithm}
\label{subsection:composite}

We describe an approach that generalizes the above approaches in order to obtain a better approximation ratio for $\ell > 2$ levels. The main idea behind this composite approach is the following: In the top-down approach, we choose a set of edges $E_{\TOP, \ell}$ that spans $T_\ell$, and then propagate this choice to levels $\ell-1, \ldots, 1$ by setting the cost of these edges to $0$. On the other hand, in the bottom-up approach, we choose a set of edges $E_{\BOT, 1}$ that spans $T_1$, which is propagated to levels $2, \ldots, \ell$ (possibly with some pruning of unneeded edges).
The idea is that for $\ell > 2$, we can choose a set of intermediate levels and propagate our choices between these levels in a top-down manner, and
to the levels lying in between them in a bottom-up manner. 

Formally, let $\mathcal{Q}=\{i_1,i_2,\dots,i_m\}$ with $1 = i_1<i_2<\dots<i_m \le \ell$ be a subset of levels sorted in increasing order. 
We first compute a Steiner tree $E_{i_m}=ST(G,T_{i_m})$ on the highest level $i_m$, which induces trees $E_{i_m + 1}, \ldots, E_{\ell}$ similar to the bottom-up approach.
Then, we set the weights of edges in $E_{i_m}$ in $G$ to zero (as in the top-down approach) and compute a Steiner tree $E_{i_{m-1}}=ST(G,T_{i_{m-1}})$ for level $i_{m-1}$ in the reweighed graph. 
Again, we can use $E_{i_{m-1}}$ to construct the trees $E_{i_{m-1}+1}, \ldots, E_{i_m - 1}$. Repeating this procedure until spanning $E_{i_1}=E_1$ results in a valid solution to the MLST problem.

Note that the top-down and bottom-up heuristics are special cases of
this approach, with $\mathcal{Q} = \{1,2,\ldots,\ell\}$ and $\mathcal{Q}
= \{1\}$, respectively.
Figure~\ref{fig:Composite} provides an illustration of how such a solution is computed in the top-down, bottom-up, and an arbitrary heuristic. Given $\mathcal{Q} \subseteq \{1,2,\ldots,\ell\}$, let $\CMP(\mathcal{Q})$ be the cost of the MLST solution returned by the composite approach over~$\mathcal{Q}$.

\begin{figure}[tb]
  \centering 
  \tikzset{every picture/.style={line width=0.75pt}} 

\begin{tikzpicture}[x=0.5pt,y=0.5pt,yscale=-1,xscale=1]

\draw  [dash pattern={on 0.84pt off 2.51pt}] (150.95,5) -- (239.5,5) -- (201.55,26) -- (113,26) -- cycle ;
\draw  [dash pattern={on 0.84pt off 2.51pt}] (150.95,41) -- (239.5,41) -- (201.55,62) -- (113,62) -- cycle ;
\draw  [dash pattern={on 0.84pt off 2.51pt}] (150.95,77) -- (239.5,77) -- (201.55,98) -- (113,98) -- cycle ;
\draw  [dash pattern={on 0.84pt off 2.51pt}] (150.95,113) -- (239.5,113) -- (201.55,134) -- (113,134) -- cycle ;
\draw  [dash pattern={on 0.84pt off 2.51pt}] (150.95,149) -- (239.5,149) -- (201.55,170) -- (113,170) -- cycle ;
\draw [color={rgb, 255:red, 255; green, 128; blue, 0 }  ,draw opacity=1 ]   (239.5,5) -- (239.5,39) ;
\draw [shift={(239.5,41)}, rotate = 270] [fill={rgb, 255:red, 255; green, 128; blue, 0 }  ,fill opacity=1 ][line width=0.75]  [draw opacity=0] (8.93,-4.29) -- (0,0) -- (8.93,4.29) -- cycle    ;

\draw [color={rgb, 255:red, 255; green, 128; blue, 0 }  ,draw opacity=1 ]   (239.5,41) -- (239.5,75) ;
\draw [shift={(239.5,77)}, rotate = 270] [fill={rgb, 255:red, 255; green, 128; blue, 0 }  ,fill opacity=1 ][line width=0.75]  [draw opacity=0] (8.93,-4.29) -- (0,0) -- (8.93,4.29) -- cycle    ;

\draw [color={rgb, 255:red, 255; green, 128; blue, 0 }  ,draw opacity=1 ]   (239.5,77) -- (239.5,111) ;
\draw [shift={(239.5,113)}, rotate = 270] [fill={rgb, 255:red, 255; green, 128; blue, 0 }  ,fill opacity=1 ][line width=0.75]  [draw opacity=0] (8.93,-4.29) -- (0,0) -- (8.93,4.29) -- cycle    ;

\draw [color={rgb, 255:red, 255; green, 128; blue, 0 }  ,draw opacity=1 ]   (239.5,113) -- (239.5,147) ;
\draw [shift={(239.5,149)}, rotate = 270] [fill={rgb, 255:red, 255; green, 128; blue, 0 }  ,fill opacity=1 ][line width=0.75]  [draw opacity=0] (8.93,-4.29) -- (0,0) -- (8.93,4.29) -- cycle    ;

\draw  [dash pattern={on 0.84pt off 2.51pt}] (305.95,5) -- (394.5,5) -- (356.55,26) -- (268,26) -- cycle ;
\draw  [dash pattern={on 0.84pt off 2.51pt}] (305.95,41) -- (394.5,41) -- (356.55,62) -- (268,62) -- cycle ;
\draw  [dash pattern={on 0.84pt off 2.51pt}] (305.95,77) -- (394.5,77) -- (356.55,98) -- (268,98) -- cycle ;
\draw  [dash pattern={on 0.84pt off 2.51pt}] (305.95,113) -- (394.5,113) -- (356.55,134) -- (268,134) -- cycle ;
\draw  [dash pattern={on 0.84pt off 2.51pt}] (305.95,149) -- (394.5,149) -- (356.55,170) -- (268,170) -- cycle ;
\draw [color={rgb, 255:red, 0; green, 0; blue, 255 }  ,draw opacity=1 ]   (268,136) -- (268,170) ;

\draw [shift={(268,134)}, rotate = 90] [fill={rgb, 255:red, 0; green, 0; blue, 255 }  ,fill opacity=1 ][line width=0.75]  [draw opacity=0] (8.93,-4.29) -- (0,0) -- (8.93,4.29) -- cycle    ;
\draw [color={rgb, 255:red, 0; green, 0; blue, 255 }  ,draw opacity=1 ]   (268,100) -- (268,134) ;

\draw [shift={(268,98)}, rotate = 90] [fill={rgb, 255:red, 0; green, 0; blue, 255 }  ,fill opacity=1 ][line width=0.75]  [draw opacity=0] (8.93,-4.29) -- (0,0) -- (8.93,4.29) -- cycle    ;
\draw [color={rgb, 255:red, 0; green, 0; blue, 255 }  ,draw opacity=1 ]   (268,64) -- (268,98) ;

\draw [shift={(268,62)}, rotate = 90] [fill={rgb, 255:red, 0; green, 0; blue, 255 }  ,fill opacity=1 ][line width=0.75]  [draw opacity=0] (8.93,-4.29) -- (0,0) -- (8.93,4.29) -- cycle    ;
\draw [color={rgb, 255:red, 0; green, 0; blue, 255 }  ,draw opacity=1 ]   (268,28) -- (268,62) ;

\draw [shift={(268,26)}, rotate = 90] [fill={rgb, 255:red, 0; green, 0; blue, 255 }  ,fill opacity=1 ][line width=0.75]  [draw opacity=0] (8.93,-4.29) -- (0,0) -- (8.93,4.29) -- cycle    ;
\draw  [dash pattern={on 0.84pt off 2.51pt}] (460.95,5) -- (549.5,5) -- (511.55,26) -- (423,26) -- cycle ;
\draw  [dash pattern={on 0.84pt off 2.51pt}] (460.95,41) -- (549.5,41) -- (511.55,62) -- (423,62) -- cycle ;
\draw  [dash pattern={on 0.84pt off 2.51pt}] (460.95,77) -- (549.5,77) -- (511.55,98) -- (423,98) -- cycle ;
\draw  [dash pattern={on 0.84pt off 2.51pt}] (460.95,113) -- (549.5,113) -- (511.55,134) -- (423,134) -- cycle ;
\draw  [dash pattern={on 0.84pt off 2.51pt}] (460.95,149) -- (549.5,149) -- (511.55,170) -- (423,170) -- cycle ;
\draw [color={rgb, 255:red, 0; green, 0; blue, 255 }  ,draw opacity=1 ]   (423,136) -- (423,170) ;

\draw [shift={(423,134)}, rotate = 90] [fill={rgb, 255:red, 0; green, 0; blue, 255 }  ,fill opacity=1 ][line width=0.75]  [draw opacity=0] (8.93,-4.29) -- (0,0) -- (8.93,4.29) -- cycle    ;
\draw [color={rgb, 255:red, 255; green, 128; blue, 0 }  ,draw opacity=1 ]   (549.5,41) -- (549.5,75) ;
\draw [shift={(549.5,77)}, rotate = 270] [fill={rgb, 255:red, 255; green, 128; blue, 0 }  ,fill opacity=1 ][line width=0.75]  [draw opacity=0] (8.93,-4.29) -- (0,0) -- (8.93,4.29) -- cycle    ;

\draw [color={rgb, 255:red, 0; green, 0; blue, 255 }  ,draw opacity=1 ]   (423,28) -- (423,62) ;

\draw [shift={(423,26)}, rotate = 90] [fill={rgb, 255:red, 0; green, 0; blue, 255 }  ,fill opacity=1 ][line width=0.75]  [draw opacity=0] (8.93,-4.29) -- (0,0) -- (8.93,4.29) -- cycle    ;
\draw [color={rgb, 255:red, 255; green, 128; blue, 0 }  ,draw opacity=1 ]   (549.5,77) -- (549.5,147) ;
\draw [shift={(549.5,149)}, rotate = 270] [fill={rgb, 255:red, 255; green, 128; blue, 0 }  ,fill opacity=1 ][line width=0.75]  [draw opacity=0] (8.93,-4.29) -- (0,0) -- (8.93,4.29) -- cycle    ;

\draw (150,180) node   {$\mathcal{Q}\ = \{1,2,3,4,5\}$};
\draw (305,180) node   {$\mathcal{Q}\ =\{1\}$};
\draw (460,180) node   {$\mathcal{Q}\ =\{1,3,4\}$};
\draw (71,159) node  [align=left] {Level 1};
\draw (71,123) node  [align=left] {Level 2};
\draw (71,87) node  [align=left] {Level 3};
\draw (71,51) node  [align=left] {Level 4};
\draw (71,15) node  [align=left] {Level 5};

\end{tikzpicture}
  \caption{Illustration of the composite heuristic for various subsets $\mathcal{Q}$, with $\ell=5$. Orange arrows pointing downward indicate propagation of edges similar to the top-down approach. Blue arrows pointing upward indicate pruning of unneeded edges, similar to the bottom-up approach.}
  \label{fig:Composite}
\end{figure}
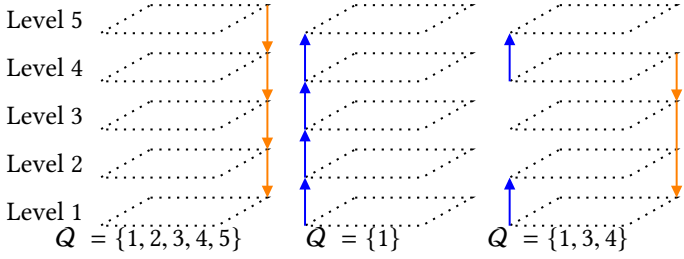

\begin{lemma}\label{lemma:cmp}
For any set~$\mathcal{Q} = \{i_1, \ldots, i_m \} \subseteq \{1,2,\ldots,\ell\}$ with $1 = i_1 < \ldots < i_m \le \ell$, we have \[\CMP(\mathcal{Q}) \le \rho \sum_{k=1}^m (i_{k+1}-1) \MIN_{i_k}\] with the convention $i_{m+1} = \ell+1$.
\end{lemma}

For example, $\mathcal{Q} = \{1,3,4\}$ with $\ell=5$ in Figure~\ref{fig:Composite} yields $\CMP(\mathcal{Q}) \le \rho\left(2 \MIN_1 + 3 \MIN_3 + 5 \MIN_4\right)$. With $\mathcal{Q} = \{1,2,3,4,5\}$, we have $\CMP(\mathcal{Q}) \le \rho\left(\MIN_1 + 2\MIN_2 + \ldots + 5\MIN_5\right)$, similar to Lemma~\ref{lemma:2level}.

\begin{proof}
The proof is similar to that of Lemma~\ref{lemma:2level}: 
when we compute $E_{i_m}$, a $\rho$-approximate Steiner tree for~$T_{i_m}$, we incur a cost of at most $\ell \rho \MIN_{i_m}$.  
This is due to the fact that the cost of $E_{i_m}$ is at most $\rho \MIN_{i_m}$, and these edges are propagated to all levels 1 through $\ell$, incurring a cost of at most $\ell \rho \MIN_{i_m}$. 
When we compute $E_{i_k}$, we incur a cost of at most $\rho \MIN_{i_k}$, and these edges are propagated to levels 1 through $i_{k+1} - 1$, incurring a cost of at most $(i_{k+1} - 1) \rho \MIN_{i_k}$.
\end{proof}
Using the trivial lower bound $\OPT\geq\sum_{i=1}^\ell \MIN_i$, we can find an upper bound for
the approximation ratio. Without loss of generality, assume
$\sum_{i=1}^\ell \MIN_i=1$, so that $\OPT \ge 1$; otherwise the edge weights can be scaled linearly. Also, since all the equations and inequalities scale linearly in $\rho$, we assume $\rho=1$. Then for a given $\mathcal{Q}$, the ratio $\frac{\CMP(\mathcal{Q})}{\OPT}$ is upper bounded by
\[\frac{\CMP(\mathcal{Q})}{\OPT}\leq
\frac{\rho \sum_{k=1}^m (i_{k+1}-1) \MIN_{i_k}}{\sum_{i=1}^\ell
  \MIN_i}= \sum_{k=1}^m (i_{k+1}-1) \MIN_{i_k}.\]



Given any $\mathcal{Q}$, we determine an approximation ratio to the MLST problem, in a way similar to the top-down and bottom-up approaches. 
We start with the following lemma:

\begin{lemma} \label{lemma:cmp-lp}
Let $c_1, c_2, \ldots, c_{\ell}$ be given non-negative real numbers with the property that $c_1 > 0$, and the nonzero ones are strictly increasing (i.e., if $i < j$ and $c_i, c_j \neq 0$, then $0 < c_i < c_j$.)
Consider the following linear program: $\max c_1y_1 + c_2y_2 + \ldots + c_{\ell}y_{\ell}$ subject to $y_1 \ge y_2 \ge \ldots \ge y_{\ell} \ge 0$, and $\sum_{i=1}^{\ell} y_i = 1$. 
Then the optimal solution has $y_1 = y_2 = \ldots = y_k = \frac{1}{k}$ for some $k$, and $y_i = 0$ for $i > k$.
\end{lemma}
\begin{proof}
Suppose that in the optimal solution, there exists some $i$ such that $y_i > y_{i+1} > 0$. If $c_{i+1} = 0$, then setting $y_1 := y_1 + y_{i+1}$ and $y_{i+1} = 0$ improves the objective function, a contradiction. If $c_{i+1} \neq 0$, then $c_i < c_{i+1}$, and it can be shown with elementary algebra that replacing $y_i$ and $y_{i+1}$ by their arithmetic mean, $\frac{y_i + y_{i+1}}{2}$, improves the objective function as well.
\end{proof}

A simple corollary of Lemma \ref{lemma:cmp-lp} is that the maximum value of the objective in the LP equals $\max\left(c_1, \frac{1}{2}(c_1 + c_2), \frac{1}{3}(c_1 + c_2 + c_3), \ldots, \frac{1}{\ell}(c_1 + \ldots + c_{\ell})\right)$. For a given $\mathcal{Q}$ (assuming $\rho = 1$), the composite heuristic on $\mathcal{Q} = \{i_1, \ldots, i_m\}$ is a $t$-approximation, where $t$ is the solution to the following simple linear program: $\max t$ subject to $t \le \sum_{k=1}^m (i_{k+1} - 1)  y_{i_k}$; 
$y_1 \ge y_2 \ge \dots \ge y_\ell \ge 0$; $\sum_{i=1}^{\ell} y_i = 1$. As this LP is of the form given in Lemma \ref{lemma:cmp-lp}, we can easily compute an approximation ratio for a given~$\mathcal{Q}$ as
\[t(Q)=\max_{m'\le m} \frac{\sum_{k=1}^{m'} (i_{k+1} - 1)}{i_{m'}}\].

For example, the corresponding objective function for $\TOP$ is $\max y_1 + 2y_2 + \ldots + \ell y_{\ell}$, and the maximum equals $\max \left(1, \frac{1}{2}(1+2), \ldots, \frac{1}{\ell} (1+2+\ldots+\ell)\right) = \frac{\ell+1}{2}$, which is consistent with Theorem~\ref{theorem:heuristics}. The corresponding objective function for $\BOT$ is $\max \ell=\ell$.

An important choice of $\mathcal{Q}$ is $\mathcal{Q} = \{1, 2, 4, \ldots, 2^m\}$ where $m = \floor{\log_2 \ell}$. 
Charikar et al.~\cite{1288137} show that this is a 4-approximation, assuming $\rho=1$. 
Indeed, according to the above formula $t=\max(1,\frac{(2-1)+(2^2-1)}{2^1},\ldots, \frac{ \sum_{i=0}^m (2^{i+1} - 1)}{2^m})=4-m/2^m \le 4$.

When $\ell \ge 2$, there are $2^{\ell-1}$ possible subsets $\mathcal{Q}$, giving $2^{\ell-1}$ possible heuristics. In particular, for $\ell=2$, the only $2^{2-1}=2$ heuristics are top-down and bottom-up (Section~\ref{subsection:heuristics}). The composite algorithm executes all possible heuristics and selects the MLST with minimum cost (denoted $\CMP$):
\[\CMP=\min_{\substack{\mathcal{Q} \subseteq \{1,\dots,\ell\} \\ 1 \in \mathcal{Q}}}\CMP(\mathcal{Q}).\]

\begin{theorem} \label{theorem:composite-lp}
For $\ell \ge 2$, the composite heuristic produces a $t_{\ell}$-approximation, where $t_{\ell}$ is the solution to the following linear program (LP):
\begin{align*}
    \max t \hspace*{9ex} \\
     \text{ subject to } \hspace*{4ex} t &\le \sum_{k=1}^{m} (i_{k+1}-1)y_{i_k} & \forall \, \mathcal{Q} = \{i_1, \ldots, i_m\} \\
                  y_i &\ge y_{i+1} & \forall \, 1 \le i \le \ell-1 \\
                  \sum_{i=1}^\ell y_{i} &= 1 \\
                  y_i &\ge 0 & \forall \, 1 \le i \le \ell
\end{align*}
\end{theorem}
\begin{proof}
Again we assume, without loss of generality, that $\rho = 1$ and that $\sum_{i=1}^{\ell} \MIN_i = 1$. Given an instance of MLST and the corresponding values $\MIN_1$, \ldots, $\MIN_{\ell}$, let $\mathcal{Q}^* = \{i_1, \ldots, i_m\}$ denote the subset of $\{1,\ldots,\ell\}$ for which the quantity $\sum_{k=1}^m (i_{k+1} - 1) \MIN_{i_k}$ is minimized. Then by Lemma \ref{lemma:cmp}, we have $\CMP \le \CMP(\mathcal{Q}^*) \le \sum_{k=1}^m (i_{k+1} - 1)\MIN_{i_k} = \hat{t}$. So for a specific instance of the MLST problem, $\CMP$ is upper bounded by $\hat{t}$, which is the minimum over $2^{\ell-1}$ different expressions, all linear combinations of $\MIN_1$, \ldots, $\MIN_{\ell}$.

As $t_{\ell}$ is the maximum of the objective over all feasible $\MIN_1$, \ldots, $\MIN_{\ell}$, we have $\hat{t} \le t_{\ell}$, so $\CMP \le t_{\ell} = t_{\ell} \cdot \OPT$ as desired.
\end{proof}

The above LP has $\ell+1$ variables and $2^{\ell-1} + 2\ell$ constraints. 
Each subset $\mathcal{Q} \in \{1,2,\ldots,\ell\}$ with $1 \in \mathcal{Q}$ corresponds to one constraint.

\begin{lemma}

The system of $2^{\ell-1}$ inequalities can be expressed in matrix form as
\[t\cdot  \boldsymbol{1}_{2^{\ell-1}\times 1} \le M_\ell\boldsymbol{y}\]
where $\boldsymbol{y}=[y_1,y_2,\cdots,y_\ell]^T$ and $M_\ell$ is a
$(2^{\ell-1}\times \ell)$-matrix that can be constructed recursively as
\[M_\ell=
\begin{bmatrix}
    P_{\ell-1}+M_{\ell-1} & \boldsymbol{0}_{2^{\ell-2}\times 1} \\
    M_{\ell-1} & \ell\cdot \boldsymbol{1}_{2^{\ell-2}\times 1} 
\end{bmatrix}
 \text{ and } P_\ell=
\begin{bmatrix}
    P_{\ell-1} & \boldsymbol{0}_{2^{\ell-2}\times 1}  \\
    \boldsymbol{0}_{2^{\ell-2}\times (\ell-1)} & \boldsymbol{1}_{2^{\ell-2}\times 1} 
\end{bmatrix} \]
starting with the $1\times 1$ matrices $M_1=[1]$ and $P_1=[1]$. 
\end{lemma}

\begin{proof}
The idea is that the rows of $P_{\ell}$ encode the largest element of their corresponding subsets - if we list the $2^{\ell-1}$ subsets of $\{1,2,\ldots,\ell\}$ in the usual ordering ($\{1\}$, $\{1,2\}$, $\{1,3\}$, $\{1,2,3\}$, $\{1,4\}$, \ldots), then $P_{i,j} = 1$ if $j$ is the largest element of the $i^{\text{th}}$ subset.

Then, given $M_{\ell-1}$ and $P_{\ell-1}$, we can construct $M_{\ell}$ by casework on whether $\ell \in \mathcal{Q}$ or not. If $\ell \not\in \mathcal{Q}$, then we build the first $2^{\ell-2}$ rows by using the previous matrix $M_{\ell-1}$, and adding 1 to the rightmost nonzero entry of each column (which is equivalent to adding $P_{\ell-1}$). If $\ell \in \mathcal{Q}$, we build the remaining $2^{\ell-2}$ rows by using the previous matrix $M_{\ell-1}$, and appending a $2^{\ell-2} \times 1$ column whose values are equal to~$\ell$.
\end{proof}

This recursively defined matrix is not central to the composite algorithm, but gives a nice way of formulating the LP above.

Solving the above LP directly is challenging for large $\ell$ due to its size. We instead use a column generation method. The idea is that solving the LP for only a subset of the constraints will produce an upper bound for the approximation ratio---larger values can be returned due to relaxing the constraints. Now, the objective would be to add ``effective'' constraints that would be most needed for getting a more accurate solution.

In our column generation, we only add one single constraint at a time. Let $\boldsymbol{Q}$ denote the set of all the constraints at the running step. Solving the LP provides a vector $\mathbf{y}$ and an upperbound for the approximation ratio $t$. Our goal is to find a new set $\mathcal{Q}_{new}=\{i_1,i_2,\dots,i_k\}$ that gives the smallest value of $\sum_{k=1}^{m} (i_{k+1}-1)y_{i_k}$ given the vector $\mathbf{y}$ from the current LP solver. We can use an ILP to find the set $\mathcal{Q}_{new}$. Specifically, we define indicator variables $\theta_{ij}$ so that $\theta_{ij}=1$ iff $i$ and $j$
are consecutive level choices for the new constraint $\mathcal{Q}_{new}$, and  $\theta_{ij}=0$ otherwise. For example, for $\mathcal{Q}_{new}=\{1,3,7\}$ with $\ell=10$ we must have $\theta_{1,3}=\theta_{3,7}=\theta_{7,11}=1$ and all other $\theta_{ij}$'s equal to zero.

 \begin{lemma}\label{Lemma: Qnew}
 Given a vector $\mathbf{y}=[y_1,\dots,y_{\ell}]$, the choice of $Q_{new}=\{i,\text{s.t., }\theta_{ij}=1\}$ from the following ILP minimizes $\sum_{k=1}^{m} (i_{k+1}-1)y_{i_k}$, where $i_k$ is the $k-$th smallest element of $Q$ and $i_{m+1}=\ell+1$.

\begin{align*}
    \min \sum_{i=1}^{\ell}\sum_{j=i+1}^{\ell+1} (j-1)\theta_{ij}y_i\\
    \text{ subject to } \hspace{5ex}
     \sum_{j>i} \theta_{ij} &\le 1,  & \sum_{i<j} \theta_{ij} &\le 1, & \forall i \in \{1,\dots,\ell\},j \in \{2,\dots,\ell+1\}\\
     \sum_{i<k} \theta_{ik} &=  \sum_{j>k} \theta_{kj} & & & \forall k\in\{2,\dots,\ell\}\\
     \sum_{1<j} \theta_{ij} &= 1, &
     \sum_{i<\ell+1}\theta_{i(\ell+1)}&= 1  \\
\end{align*}

\end{lemma}

\begin{proof}
Using the indicator variables, $\sum_{k=1}^{m} (i_{k+1}-1)y_{i_k}$ can be expressed as $\sum_{i=1}^{\ell}\sum_{j=i+1}^{\ell+1} (j-1)\theta_{ij}y_i$ because $\theta_{i_k,i_{k+1}}=1$ and other $\theta_{ij}$'s are zero. In the above formulation, the first constraint indicates that, for every given $i$ or $j$, at most one $\theta_{ij}$ is equal to one. The second constraint indicates that, for a given $k$, if $\theta_{ik}=1$ for some $i$, then there is also a $j$ such that $\theta_{kj}=1$. In other words, the result determines a proper choice of levels by ensuring continuity of $[i,j]$ intervals. The last constraints guarantees that levels $1$ is chosen for $\mathcal{Q}_{new}$.
\end{proof}

Lemma \ref{Lemma: Qnew} allows for a column generation technique to solve the LP program of Theorem \ref{theorem:composite-lp} for computing the approximation ratio of $\CMP$ algorithm. We initially start with an empty set $\mathcal{Q}=\{\}$ and $\mathbf{y}=[1,\frac{1}{2},\dots,\frac{1}{\ell}].$ Then using Lemma \ref{Lemma: Qnew}, we find $Q_{new}$ and add it to the constraint set $\mathcal{Q}$. We repeat the process until $Q_{new}$ already belongs to $\mathcal{Q}$. Without the column generation technique, we could only comfortably go up to $\ell=22$ levels, however, the new techniques allows us to solve for much larger values of $\ell$. In the following, we report the results up to 100 levels. 


\begin{theorem}\label{corollary:composite}
For any $~\ell=2,\dots,100$, the MLST returned by the composite algorithm yields a solution whose cost is no worse than $t_{\ell} \OPT$, where the values of $t_{\ell}$ are listed in Figure~\ref{fig:lvr}.
\end{theorem}


\begin{figure}[tb]
\centering
\begin{subfigure}{.55\textwidth}
\includegraphics[width=0.9\textwidth]{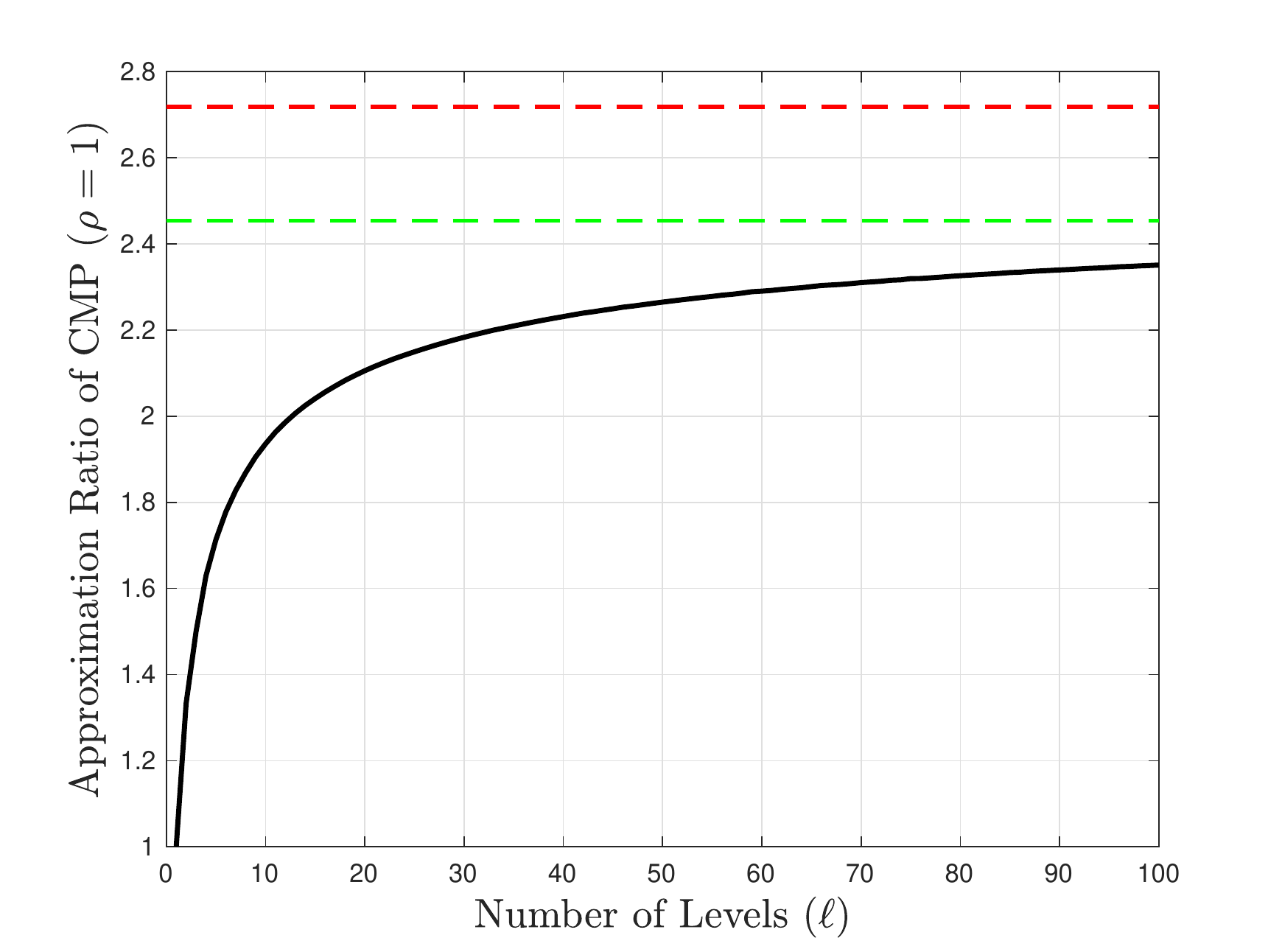}
\end{subfigure}\hfill
{\small 
\begin{subtable}{.2\textwidth}
\begin{tabular}{rr}
$\ell$ & $t_{\ell}$ \\\hline
1 & $\rho$ \\
2 & 1.333$\rho$\\
3 & 1.500$\rho$\\
4 & 1.630$\rho$\\
5 & 1.713$\rho$\\
6 & 1.778$\rho$\\
7 & 1.828$\rho$\\
8 & 1.869$\rho$\\
9 & 1.905$\rho$\\
10 & 1.936$\rho$\\
11 & 1.963$\rho$\\
\end{tabular}
\end{subtable}
\begin{subtable}{.2\textwidth}
\begin{tabular}{rr}
$\ell$ & $t_{\ell}$ \\\hline
12 & 1.986$\rho$\\
13 & 2.007$\rho$\\
14 & 2.025$\rho$\\
15 & 2.041$\rho$\\
16 & 2.056$\rho$\\
17 & 2.070$\rho$\\
18 & 2.083$\rho$\\
19 & 2.094$\rho$\\
20 & 2.106$\rho$\\
50 & 2.265$\rho$\\
100 & 2.351$\rho$\\
\end{tabular}
\end{subtable}
}
\caption{Approximation ratios for the composite algorithm for $\ell=1,\dots,100$ (black curve), compared to the ratio $t=e\rho$ (red dashed line) guaranteed by the algorithm of Charikar et al.~\cite{1288137} and $t=2.454\rho$ (green dashed line) guaranteed by the algorithm of Karpinski et al.~\cite{Karpinski2005}. The table to the right lists the exact values for the ratio $t/\rho$.}
\label{fig:lvr}
\end{figure}



Neglecting the factor~$\rho$ for now (i.e., assuming $\rho=1$), the approximation ratio $t=3/2$ for $\ell=3$ is slightly better than the ratio of ${(5+4\sqrt{2})}/{7} +\varepsilon \approx 1.522 +\varepsilon$ guaranteed by Xue et al.~\cite{Xue2001} for the Euclidean case. 
(The additive constant $\varepsilon$ in their ratio stems from using Arora's PTAS as a subroutine for Euclidean \STP, which corresponds to the multiplicative constant~$\rho$ for using an \STP algorithm as a subroutine for MLST.)
Recall that an improvement for $\ell=3$ was posed as an open problem by Karpinski et al.~\cite{Karpinski2005}. 
Also, for each of the cases $4 \leq \ell \leq 100$ 
our results in Theorem~\ref{corollary:composite} 
improve the approximation ratios of $e\rho\approx 2.718\rho$ and $2.454\rho$ guaranteed by Charikar et al.~\cite{1288137} and by Karpinski et al.~\cite{Karpinski2005}, respectively.
The graph of the approximation ratio of the composite algorithm (see
Figure~\ref{fig:lvr}) for $\ell=1,\dots,100$ suggests that it will stay
below $2.454\rho$ for values of~$\ell$ much larger than~$100$.

The obvious disadvantage is that computing $\CMP$ involves $2^{\ell-1}$ different heuristics, requiring $2^{\ell-2}(\ell+1)$ \STP computations, which is not feasible for large $\ell$. In the following, we show that we can achieve the same approximation guarantee with at most $2\ell$ \STP computations.  

\begin{theorem}
Given an instance of the MLST problem, a specific subset $\mathcal{Q}^* \subseteq \{1,2,\ldots,\ell\}$ (with $1 \in \mathcal{Q}^*$) can be found through $\ell$ \STP computations, such that running the composite heuristic on $\mathcal{Q}^*$ is a $t_{\ell}$-approximation.
\end{theorem}
\begin{proof}
Given a graph $G=(V,E)$ with cost function $c$, and terminal sets
$T_{\ell}\subset T_{\ell-1}\subset \dots \subset T_1\subseteq V$, compute a
Steiner tree on each level and set $\MIN_i=c(ST(G,T_i))$. Again, assume w.l.o.g. that $\sum_{i=1}^{\ell} \MIN_i = 1$, which can be done by scaling the edge weights appropriately after computing the minimum Steiner trees.

Since $\boldsymbol{y}=[\MIN_1,\dots,\MIN_{\ell}]^T$ and $t = \min M_{\ell} \boldsymbol{y}$ is a feasible, but not necessarily optimal, solution to the LP for computing the approximation ratio~$t_{\ell}$,
we have $\min M_{\ell} \boldsymbol{y} = t \le t_{\ell}$. Let $q \in \{1,2,\ldots,2^{\ell-1}\}$ be the row of $M_{\ell}$ whose dot product with $\boldsymbol{y}$ is minimized (i.e. equals $t$), and let $\mathcal{Q}^*$ be the subset of levels corresponding to the $q^{\text{th}}$ row of $M_{\ell}$, which can be obtained from the non-zero entries of the $q^{\text{th}}$ row. Then on this specific subset $\mathcal{Q}^{*}$ of levels, we have $\CMP(Q^{*}) \le t \le t_{\ell} \le t_{\ell} \OPT$. Then the optimal choice of $Q^{*}$ given $\boldsymbol{y}=[\MIN_1,\dots,\MIN_{\ell}]^T$ can be obtained very efficiently using the ILP of the column generation.
%
\end{proof}

Note that the number of \STP computations is reduced from $2^{\ell-2}(\ell+1)$ to at most $2\ell$. The solution with cost $\CMP(Q^{*})$ does not necessarily have the same cost as the solution with cost $\CMP$, however, the solution returned is still at most $t_{\ell}$ times the optimum. It is worth noting that the analyses in this section did not assume that the computed edge sets are trees, only that the edge sets are nested and that the cost of a solution is the sum of the costs over all levels.

\section{Integer Linear Programming (ILP) Formulations}
\label{section:exact_algorithm}
In this section, we discuss several different ILP formulations for the MLST problem.

\subsection{ILP Based on Cuts} \label{subsection:ilp1}
This is a standard ILP formulation for the (single-level) Steiner tree problem. Recall that $T \subseteq V$ is the set of terminals. Given a cut $S \subseteq V$, let $\delta(S) = \{uv \in E \mid u \in S, v \not\in S\}$ denote the set of all edges that have exactly one endpoint in $S$. Given an undirected edge $uv \in E$, let $x_{uv} = 1$ if $uv$ is present in the solution, 0 otherwise. An ILP formulation for \STP is as follows:

\begin{align*}
\text{Minimize } \sum \limits_{uv \in E} c(u,v) \cdot x_{uv} \text{ subject to} \hspace{3ex} \\
\sum_{uv \in \delta(S)} x_{uv} &\ge 1 & \forall \, S \subset V; S \cap T \neq \emptyset; S \cap T \neq T \\
0 \le x_{uv} &\le 1 & \forall uv \in E
\end{align*}

The cut-based formulation generalizes to $\ell$ levels naturally. Let $x_{uv}^i = 1$ if edge $uv$ is present on level $i$, and 0 otherwise. We constrain that the graph on level $i$ is a subgraph of the graph on level $i-1$ by requiring $x_{uv}^i \le x_{uv}^{i-1}$ for all $2 \le i \le \ell$ and for all $uv \in E$. Then a cut-based formulation for the MLST problem is as follows:

\begin{align*}
\text{Minimize } \sum \limits_{i=1}^{\ell} \sum \limits_{uv \in E} c(u,v) \cdot x_{uv}^i \text{ subject to} \hspace{3ex} \\
\sum_{uv \in \delta(S)} x_{uv}^i &\ge 1 & \forall \, S \in V; S \cap T \neq \emptyset,T; 1 \le i \le \ell \\
x_{uv}^i &\le x_{uv}^{i-1} & \forall uv \in E; 2 \le i \le \ell \\
0 \le x_{uv}^i &\le 1 & \forall uv \in E; 1 \le i \le \ell
\end{align*}

The number of variables is $O(\ell|E|)$, however, the number of constraints is $O(\ell \cdot 2^{|V|})$.


 \subsection{ILP Based on Multi-Commodity Flow} \label{subsection:ilp2}
We recall here the well-known undirected flow formulation for \STP ~\cite{Polzin2001, Balakrishnan1994}.
Let $s \in T$ be a fixed terminal node, the \emph{source}. Given an edge $uv \in E$, the indicator variable $x_{uv}$ equals 1 if the edge $uv$ is present in the solution and 0 otherwise. This formulation sends $|T|-1$ unit commodities from the source $s$ to each terminal in $T - \{s\}$. The variable $f_{uv}^p$ denotes the (integer) flow from $u$ to $v$ of commodity $p$. A multi-commodity ILP formulation for \STP is as follows:

\begin{align*}
    \text{Minimize } \sum \limits_{uv \in E} c(u,v) \cdot x_{uv} \text{ subject to} \hspace{3ex} \\
\sum \limits_{vw \in E} f_{vw}^p - \sum \limits_{uv \in E} f_{uv}^p & = 
\begin{cases}
    1       & \quad \text{if } v=s\\
    -1  & \quad \text{if } v=p \\
    0  & \quad \text{else}\\
  \end{cases} & \forall \, v \in V\\
0 \leq f_{uv}^p &\le 1 & \forall \, uv \in E \\
0 \leq f_{vu}^p &\le 1 & \forall \, uv \in E \\
0 \le x_{uv} &\le 1 & \forall \,uv \in E
\end{align*}
To generalize to the MLST problem, 
we add the linking constraints ($x_{uv}^i \le x_{uv}^{i-1}$ as before) and enforce the flow constraints on levels $1, \ldots, \ell$.

\subsection{ILP Based on Single-Commodity Flow} \label{subsection:ilp3}
\STP can also be formulated using a single commodity flow, instead of multiple commodities. Here, we will assume the input graph is directed, by replacing each edge $uv$ with directed edges $(u,v)$ and $(v,u)$ of the same cost. As before, $f_{uv}$ denotes the flow from $u$ to $v$:

\begin{align*}
\text{Minimize } \sum \limits_{(u,v) \in E} c(u,v) \cdot x_{uv} \text{ subject to} \hspace{3ex} \\
\sum \limits_{(v,w) \in E} f_{vw} - \sum \limits_{(u,v) \in E} f_{uv} & = 
\begin{cases}
    |T|-1       & \quad \text{if } v=s\\
    -1  & \quad \text{if } v \in T \setminus \{s\} \\
    0  & \quad \text{else}\\
  \end{cases} & \forall \, v \in V\\
0 \leq f_{uv} & \leq (|T|-1) \cdot x_{uv} & \forall \, (u,v) \in E \\
0 \le x_{uv} &\le 1 & \forall \,(u,v) \in E
\end{align*}
To generalize to multiple levels, we add the linking constraints $x_{uv}^i \ge x_{uv}^{i-1}$ as before. Let $f_{uv}^i$ denote the integer flow along edge $(u,v)$ on level $i$. Let $s \in T_{\ell}$ be a source terminal on the top level $T_{\ell}$. Then the MLST problem can be formulated using single-commodity flows on each level:
\begin{align*}
\text{Minimize } \sum \limits_{i=1}^{\ell} \sum \limits_{(u,v) \in E} c(u,v) x_{uv}^i \text{ subject to} \\
\sum \limits_{(v,w) \in E} f_{vw}^i - \sum \limits_{(u,v) \in E} f_{uv}^i & = 
\begin{cases}
    |T_i|-1       & \quad \text{if } v=s\\
    -1  & \quad \text{if } v \in T_i \setminus \{s\} \\
    0  & \quad \text{else}\\
  \end{cases} & \forall \, v \in V; 1 \le i \le \ell \\
x_{uv}^{i} & \ge x_{uv}^{i-1} & \forall \, (u,v) \in E; 2 \le i \le \ell \\
0 \leq f_{uv}^i & \leq (|T_i|-1) \cdot x_{uv}^i & \forall \, (u,v) \in E; 1 \le i \le \ell \\
0 \le x_{uv}^i &\le 1 & \forall \, (u,v) \in E; 1 \le i \le \ell
\end{align*}
The number of variables is $O(\ell |E|)$ and the number of constraints is $O(\ell (|E|+|V|))$. In Section~\ref{subsection:ilp4}, we reduce the number of variables and constraints to $O(|E|)$ and $O(|E|+|V|)$, respectively.

\subsection{A Smaller ILP Based on Single-Commodity Flow} \label{subsection:ilp4}
We can simplify the flow-based ILP in Section~\ref{subsection:ilp3} so that the number of variables is $O(|E|)$ and the number of constraints is $O(|E|+|V|)$. This is done by only enforcing the single-commodity flow on the bottom level. Let $L(v)$ denote the highest level that $v$ is a terminal in, i.e. if $v \in T_i$ and $v \not\in T_{i+1}$, then $L(v) = i$. If $v \not\in T_1$, then $L(v) = 0$. For each directed edge $(u,v) \in E$, let $x_{uv} = 1$ if $(u,v)$ appears on the bottom level (level 1) in the solution, and $x_{uv} = 0$ otherwise. Let $y_{uv}$ denote the highest level that $(u,v)$ appears in, i.e. $y_{uv} = i$ if $(u,v)$ is present on level $i$ but not on level $i+1$, and $y_{uv} = 0$ if $(u,v)$ is not present anywhere. The variables $y_{uv}$ indicate the number of times we pay the cost of edge $(u,v)$ in the solution. Let $N(v) = \{u \in V \mid (u,v) \in E\}$ denote the neighborhood of $v$. As in Section \ref{subsection:ilp3}, let $f_{uv}$ denote the flow along directed edge $(u,v)$, and let $s \in T_{\ell}$ be the source. A reduced ILP formulation is as follows:

\begin{align}
\text{Minimize } \sum_{uv \in E} c(u,v) (y_{uv} + y_{vu}) \text{ subject to} \\
\sum_{(v,w) \in E} f_{vw} - \sum_{(u,v) \in E} f_{uv} &= \begin{cases} |T_1| - 1 & v = s \\
-1 & v \in T_1 - \{s \} \\
0 & \text{ otherwise}
\end{cases} & \forall v \in V \label{eqn:ilp4-2} \\
0 \le f_{uv} &\le (|T_1| - 1)x_{uv} & \forall uv \in E \label{eqn:ilp4-3} \\
\sum_{u \in N(v)} x_{uv} &\le 1 &\forall v \in V \label{eqn:ilp4-4} \\
x_{uv} \le y_{uv} &\le \ell x_{uv} & \forall uv \in E  \label{eqn:ilp4-5} \\
\sum_{u \in N(v) - \{w\}} x_{uv} &\ge x_{vw} & \forall vw \in E, v \neq s \label{eqn:ilp4-6} \\
\sum_{u \in N(v) - \{w\}} y_{uv} &\ge y_{vw} & \forall vw \in E, v \neq s \label{eqn:ilp4-7} \\
\sum_{u \in N(v)} y_{uv} &\ge L(v) & \forall v \in T_1 - \{s\} \label{eqn:ilp4-8} \\
0 \le x_{uv} &\le 1 & \forall uv \in E \label{eqn:ilp4-9}
\end{align}

Constraints (\ref{eqn:ilp4-2}) and (\ref{eqn:ilp4-3}) are the same as before, but we only enforce the flow constraint on the bottom level (level 1). Constraint (\ref{eqn:ilp4-4}) enforces that each vertex has at most one incoming edge in the solution. 

Constraint (\ref{eqn:ilp4-5}) ensures that if $x_{uv} = 0$, then we do not incur cost for edge $(u,v)$, and so $y_{uv} = 0$. Otherwise, if $x_{uv} = 1$, then we pay for edge $(u,v)$ between 1 and $\ell$ times, or $1 \le y_{uv} \le \ell$.

Constraint (\ref{eqn:ilp4-6}) ensures that if $(v,w)$ is an edge in the solution, and $v$ is not the root, then there is some $u \neq w$ for which $(u,v)$ is an edge in the solution. Combined with (\ref{eqn:ilp4-4}), this implies that $v$ has exactly one incoming edge, and it is not the edge $(w,v)$.

Constraint (\ref{eqn:ilp4-7}) ensures that if edge $(v,w)$ appears on levels $1, \ldots, i$ (i.e. $y_{uv} = i$), and $v$ is not the root, then the sum over all neighbors $u$ of $v$ (other than $w$) of $y_{uv}$ is at least $i$. As $v$ has exactly one incoming edge $(u,v)$ and $u \neq w$, constraint (\ref{eqn:ilp4-7}) combined with (\ref{eqn:ilp4-4}) and (\ref{eqn:ilp4-5}) together imply that $v$ has exactly one incoming edge, and its level is greater than or equal to $y_{vw}$.

Constraint (\ref{eqn:ilp4-8}) ensures that if $v$ is any terminal besides the root, then the sum over all neighbors $u$ of $y_{uv}$ is at least its level, $L(v)$. Combined with (\ref{eqn:ilp4-4}) and (\ref{eqn:ilp4-5}), this implies that every non-root terminal has exactly one incoming edge that appears on at least $L(v)$ levels.

Constraint (\ref{eqn:ilp4-9}) ensures that all $x_{uv}$ variables are binary, which implies $0 \le y_{uv} \le \ell$ for all variables $y_{uv}$. Given a solution to the above ILP, the graph $G_1$ is such that $uv \in G_1$ if $x_{uv} = 1$ or $x_{vu} = 1$. More generally, $uv \in G_i$ if $y_{uv} \ge i$ or $y_{vu} \ge i$.

\begin{lemma} \label{lemma:ilp4-1}
In the \emph{optimal} solution to the above ILP, the graph $G_1 = (V,E_1)$ with $E_1 = \{uv \in E \mid x_{uv} = 1 \text{ or } x_{vu} = 1\}$ is a Steiner tree spanning all terminals $T_1$.
\end{lemma}
\begin{proof}
We show that (i) $G_1$ contains no cycle, (ii) there exists a path in $G_1$ from the source $s$ to every terminal $v \in T_1$, and (iii) $G_1$ is connected.

\begin{enumerate}[label=(\roman*)]
\item Assume otherwise that $G_1$ contains a cycle $C$. Such a cycle contains directed edges oriented in the same direction, otherwise, there exist vertices $u$, $v$, $w$ along the cycle such that $x_{uv} = x_{wv} = 1$, which violates (\ref{eqn:ilp4-4}). Additionally, such a cycle cannot have any ``incoming'' edges (edges $(u,v)$ where $u \not\in C$ and $v \in C$), as this violates (\ref{eqn:ilp4-4}) as well.

If $C$ contains $s$, then removing its preceding edge $(v,s)$ and reducing all flows on $C$ by $f_{vs}$ yields a feasible solution with lower cost, contradicting optimality. If $C$ does not contain $s$, but contains some terminal $v \in T_1$, then since there are no incoming edges into $C$, we cannot satisfy the flow constraint on $v$. If $C$ contains no terminal, then removing $C$ along with its edge flows gives a solution with lower cost.

\item Consider an arbitrary terminal $v \in T_1$ with $v \neq s$. Then using previous arguments, $v$ has exactly one incoming edge $(u,v)$ (whose level is at least $L(v)$). If $u = s$, we are done. Otherwise, we continue this process until we reach the source $s$. Note that continuing this process does not revisit a vertex, as $G_1$ contains no cycle.

\item Since there exists a directed path from $s$ to each terminal $v \in T_1$, then all terminals are in the same connected component in $G_1$. If there exist other connected components in $G_1$, then removing them and setting flows to zero yields a solution with lower cost.
\end{enumerate}
\end{proof}

\begin{lemma} \label{lemma:ilp4-2}
In the \emph{optimal} solution to the above ILP, the graph $G_i = (V, E_i)$ with $E_i = \{uv \in E \mid y_{uv} \ge i \text{ or } y_{vu} \ge i\}$ is a Steiner tree spanning all terminals $T_i$.
\end{lemma}
\begin{proof}
The graph $G_i$ is necessarily a subgraph of $G_1$, since $y_{uv} \ge i$ or $y_{vu} \ge i$ implies $x_{uv} \ge 1$ or $x_{vu} \ge 1$ by (\ref{eqn:ilp4-5}). Consider some terminal $v \in T_i$, $v \neq s$. By constraint (\ref{eqn:ilp4-8}), there is exactly one incoming edge $(u,v)$ to $v$ such that $y_{uv} \ge L(v)$. Applying similar arguments using constraint (\ref{eqn:ilp4-7}),  we will eventually reach the root via a path, all of whose edges appear at least $L(v)$ times.
\end{proof}

\begin{theorem}
The optimal solution to the above ILP with cost $\OPT_{ILP}$ yields the optimal MLST solution.
\end{theorem}
\begin{proof}
The optimal MLST solution whose cost is $\OPT$ is a feasible solution to the ILP, as we can set the $x_{uv}$, $y_{uv}$, and $f_{uv}$ variables accordingly, so $\OPT_{ILP} \le \OPT$. By Lemmas \ref{lemma:ilp4-1} and \ref{lemma:ilp4-2}, the optimal solution $\OPT_{ILP}$ gives a feasible solution to the MLST problem, so $\OPT \le \OPT_{ILP}$. Then $\OPT = \OPT_{ILP}$.
\end{proof}

The number of flow variables is $2|E|$ (where $|E|$ is the number of edges in the input graph), and the total number of variables is $O(|E|)$. The number of flow constraints is $O(|V|)$ and the total number of constraints is $O(|V|+|E|)$. Additionally, the integrality constraints on the flow variables $f_{uv}$ as well as the variables $y_{uv}$ may be dropped without affecting the optimal solution.


\section{Experimental Results}
\label{section:experimental_results}


\myparagraph{Graph Data Synthesis.} 
The graphs we used in our experiment are synthesized from graph
generation models. In particular, we used three random network
generation models: Erd\H{o}s--Renyi (ER) \cite{erdos1959random},
Watts--Strogatz (WS) \cite{watts1998collective}, and Barab\'{a}si--Albert (BA) \cite{barabasi1999emergence}. These networks are very well studied in the literature \cite{newman2003structure}.

The Erd\H{o}s--Renyi model, $\textsc{ER}(n,p)$, assigns an edge to every possible pair among $n = |V|$ vertices with probability $p$, independently of other edges. It is well-known that an instance of $\textsc{ER}(n,p)$ with $p=(1+\varepsilon)\frac{\ln n}{n}$ is almost surely connected for $\varepsilon>0$ \cite{erdos1959random}. For our experiment, $n = 50, 100, 150, \ldots, 500$, and $\varepsilon = 1$.


The Watts--Strogatz model, $\textsc{WS}(n,K,\beta)$, initially creates a ring lattice of constant degree $K$, and then rewires each edge with probability $0\leq \beta \leq 1$ while avoiding self-loops or duplicate edges. Interestingly, the Watts--Strogatz model generates graphs that have the small-world property while having high clustering coefficient \cite{watts1998collective}. In our experiment, the values of $K$ and $\beta$ are equal to $6$ and $0.2$ respectively.

The Barabási–Albert model, $\textsc{BA}(m_0,m)$, uses a preferential attachment mechanism to generate a growing scale-free network. The model starts with a graph of $m_0$ vertices. Then, each new vertex connects to $m\leq m_0$ existing nodes with probability proportional to its instantaneous degree. The BA model generates networks with power-law degree distribution, i.e. few vertices become hubs with extremely large degree \cite{barabasi1999emergence}. This model is a network growth model. In our experiments, we let the network grow until a desired network size $n$ is attained. We vary $m_0$ from $10$ to $100$ in our experiments. We keep the value of $m$ equal to $5$.

For each generation model, we generate graphs on size $|V| = 50, 100, 150, \ldots, 500$. On each graph instance, we assign integer edge weights $c(e)$ randomly and uniformly between 1 and 10 inclusive. We only consider connected graphs in our experiment. Computational challenges of solving an ILP limit the size of the graphs to a few hundred in practice.

\myparagraph{Selection of Levels and Terminals.} For each generated graph, we generated MLST instances with $\ell = 2, 3, 4, 5$ levels.
We adopted two strategies for selecting the terminals on the $\ell$ levels: \emph{linear} vs. \emph{exponential}.
In the linear case, we select the terminals on each level by randomly sampling $\lfloor {|V| \cdot (\ell-i+1)}/(\ell+1) \rfloor$ vertices on level $i$ so that the size of the terminal sets decreases linearly. As the terminal sets are nested, $T_i$ can be selected by sampling from $T_{i-1}$ (or from $V$ if $i=1$).
In the exponential case, we select the terminals on each level by randomly sampling $\lfloor {|V|}/{2^{\ell-i+1}} \rfloor$ vertices so that the size of the terminal sets decreases exponentially by a factor of 2.

To summarize, a single instance of an input to the MLST problem is characterized by four parameters: network generation model 
$\textsc{NGM}\in\{\textsc{ER,WS,BA}\}$, number of vertices $|V|$, number of levels $\ell$, and the terminal selection method $\textsc{TSM}\in\{\textsc{Linear,Exponential}\}$. Since each instance of the experiment setup involves randomness at different steps, we generated 5 instances for every choice of parameters (e.g., WS, $|V| = 100$, $\ell=5$, \textsc{Linear}).

\myparagraph{Algorithms and Outputs.}
We implemented the bottom-up, top-down, and composite heuristics described in Section~\ref{section:approximation_algorithms}.

For evaluating the heuristics, we also implemented all ILPs described in Section~\ref{section:exact_algorithm} using CPLEX 12.6.2 as an ILP solver. The ILP described in Section~\ref{subsection:ilp4} works very well in practice. Hence, we have used this ILP for our experiment.
The model of the HPC system we used for our experiment is Lenovo NeXtScale nx360 M5. It is a distributed system; the models of the processors in this HPC are Xeon Haswell E5-2695 Dual 14-core and Xeon Broadwell E5-2695 Dual 14-core. The speed of a processor is 2.3 GHz. There are 400 nodes each having 28 cores. Each node has 192 GB memory. The operating system is CentOS 6.10.




For each instance of the MLST problem, we compute the costs of the MLST from the
ILP solution (\OPT), the bottom-up solution (\BOT), the top-down
solution (\TOP), the composite heuristic (\CMP), and the guaranteed
performance heuristic ($\CMP(\mathcal{Q^*})$ where $\mathcal{Q}^{*}$ is chosen suitably).
The four heuristics involve a (single-level) \STP subroutine; we used both the 2-approximation algorithm of Gilbert and Pollak~\cite{Gilbert1968}, as well as the flow formulation described in Section~\ref{subsection:ilp4} which solves \STP optimally. The purpose of this is to assess whether solving \STP optimally significantly improves the approximation ratio.

After completing the experiment, we compared the results of the heuristics with exact solutions. 
We show the performance ratio for each heuristic (defined as the heuristic cost divided by \OPT), and how the ratio depends on the experiment parameters (number of levels $\ell$, terminal selection method, number of vertices $|V|$).
We record the number of \STP computations involved for the guaranteed performance heuristic ($\CMP(\mathcal{Q}^{*})$) (note that this equals $\ell + |\mathcal{Q}^{*}|$). Finally, we discuss the running time of the ILP we have used in our experiment. All box plots shown below show the minimum, interquartile range (IQR) and maximum, aggregated over all instances using the parameter being compared.


\myparagraph{Results.}
%
We found that the four heuristics perform very well in practice using the 2-approximation algorithm as a (single-level) \STP subroutine, and that using an exact \STP subroutine did not significantly improve performance. Hence, we only discuss the results that use the 2-approximation as a subroutine.

Figure~\ref{BoxPlots_L} shows the performance of the four heuristics compared with the optimal solution as a function of $\ell$.
As expected, the performance of the heuristics gets slightly worse as $\ell$ increases. The bottom-up approach had the worst performance, while the composite heuristic performed very well in practice.

 Figure~\ref{BoxPlots_ND} shows the performance of the four heuristics compared with the optimal solution as a function of terminal selection, either \textsc{linear} or \textsc{exponential}.
 Overall, the heuristics performed slightly worse when the sizes of the terminal sets decrease exponentially.

\begin{figure}
\begin{minipage}{\textwidth}
    \centering
    \begin{subfigure}[b]{0.32\textwidth}
        \includegraphics[width=\textwidth]{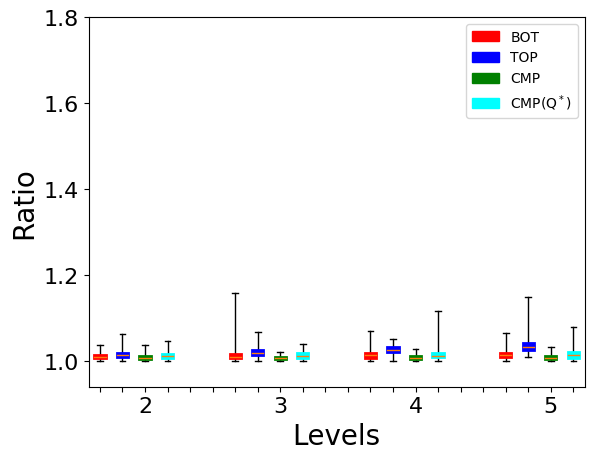}
        \caption{Barab{\'a}si--Albert}
    \end{subfigure}
    ~ 
    \begin{subfigure}[b]{0.32\textwidth}
        \includegraphics[width=\textwidth]{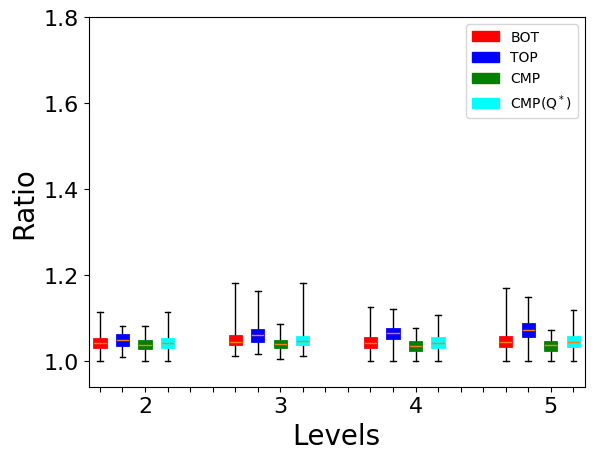}
        \caption{Erd\H{o}s--R{\'e}nyi}
    \end{subfigure}
    ~
    \begin{subfigure}[b]{0.32\textwidth}
        \includegraphics[width=\textwidth]{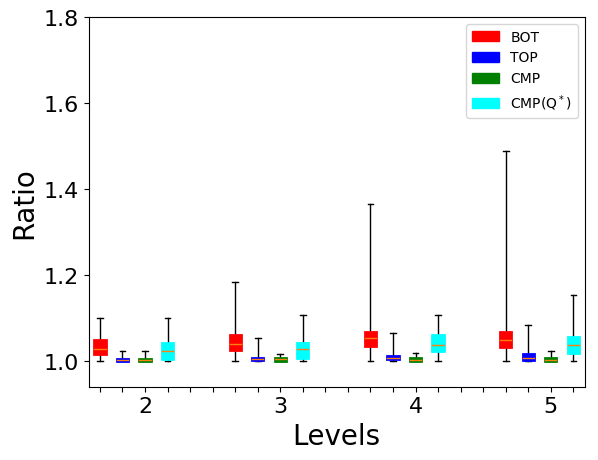}
        \caption{Watts--Strogatz}
    \end{subfigure}
    \caption{Performance of \BOT, \TOP, \CMP, and \CMPS w.r.t.\ the number~$\ell$ of levels using the 2-approximation for \STP as a subroutine.} 
    \label{BoxPlots_L}
\end{minipage}

\vspace{4ex}

\begin{minipage}{\textwidth}
    \centering
    \begin{subfigure}[b]{0.32\textwidth}
        \includegraphics[width=\textwidth]{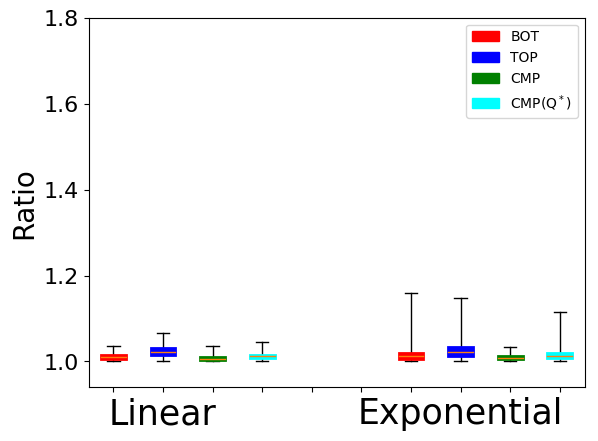}
        \caption{Barab{\'a}si--Albert}
    \end{subfigure}
    ~ 
    \begin{subfigure}[b]{0.32\textwidth}
        \includegraphics[width=\textwidth]{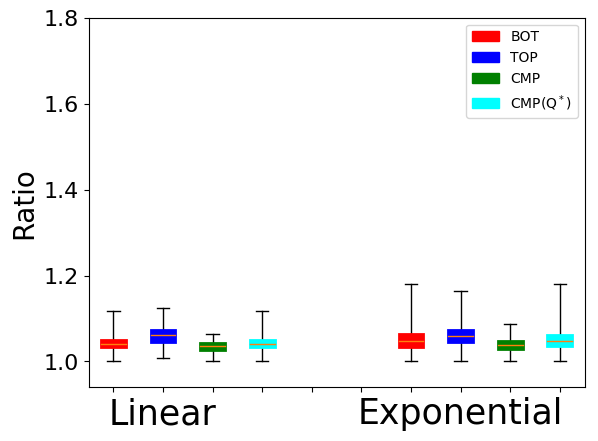}
        \caption{Erd\H{o}s--R{\'e}nyi}
    \end{subfigure}
    ~
    \begin{subfigure}[b]{0.32\textwidth}
        \includegraphics[width=\textwidth]{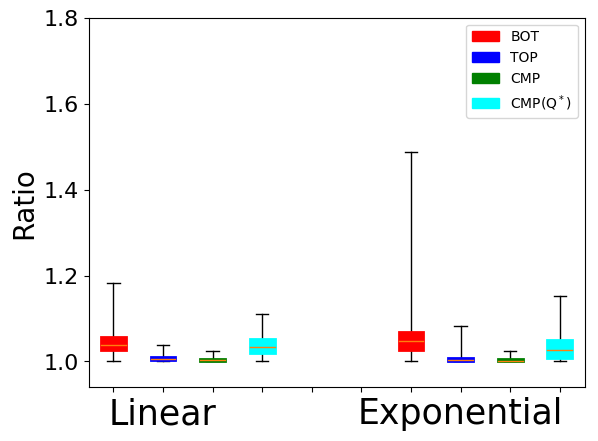}
        \caption{Watts--Strogatz}
    \end{subfigure}
    \caption{Performance of \BOT, \TOP, \CMP, and \CMPS w.r.t.\ the terminal selection method using the 2-approximation for \STP as a subroutine.} 
    \label{BoxPlots_ND}
\end{minipage}
\end{figure}

Figures ~\ref{ER_N} through \ref{BA_N} show the performance of the heuristics compared with the optimal solution, as a function of the number of vertices $|V|$.
The minimum, average, and maximum values for ``Ratio'' are aggregated over all instances of $|V|$ vertices (terminal selection, number of levels $\ell$, 5 instances for each). Due to space, we omit the bottom-up (BU) heuristic here, which tends to be comparable or slightly worse in performance than the top-down (TD) heuristic. Again, the composite (CMP) has the best ratio as it selects the best over all $2^{\ell-1}$ possible solutions; top-down and \CMPS were comparable.



\begin{figure}
\begin{minipage}{\textwidth}
    \centering
    \begin{subfigure}[b]{0.31\textwidth}
        \includegraphics[width=\textwidth]{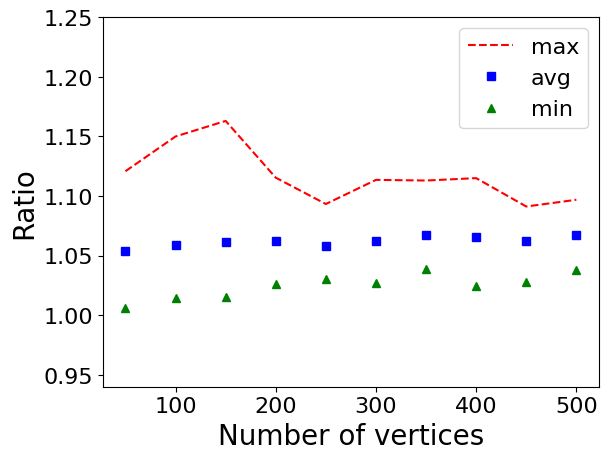}
        \caption{Top-down}
    \end{subfigure}
    ~
    \begin{subfigure}[b]{0.31\textwidth}
        \includegraphics[width=\textwidth]{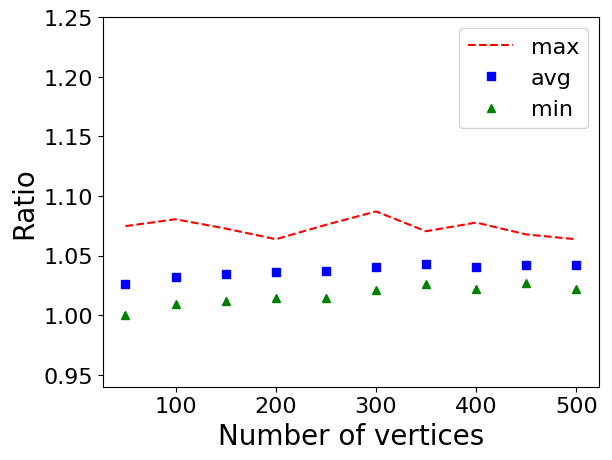}
        \caption{Composite}
    \end{subfigure}
    ~
    \begin{subfigure}[b]{0.31\textwidth}
        \includegraphics[width=\textwidth]{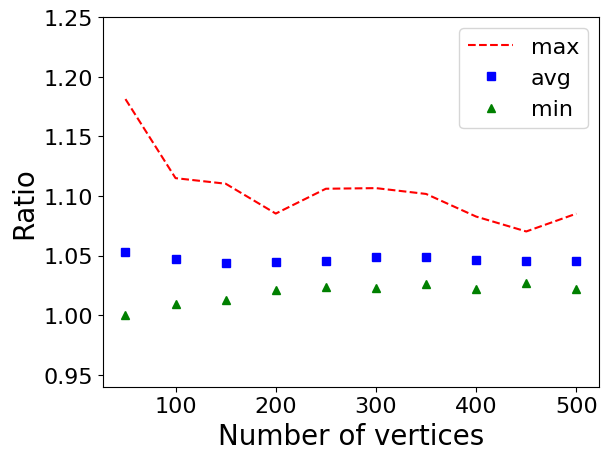}
        \caption{\CMPS}
    \end{subfigure}

    \caption{Performance of \TOP, \CMP, and \CMPS on Erd\H{o}s--R{\'e}nyi graphs using the 2-approximation for \STP as a subroutine.} 
    \label{ER_N}
\end{minipage}

\vspace{5ex}

\begin{minipage}{\textwidth}
    \centering
    \begin{subfigure}[b]{0.31\textwidth}
        \includegraphics[width=\textwidth]{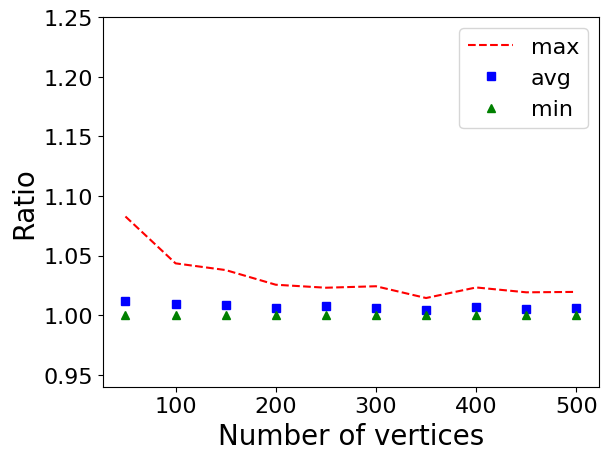}
        \caption{Top-down}
    \end{subfigure}
    ~
    \begin{subfigure}[b]{0.31\textwidth}
        \includegraphics[width=\textwidth]{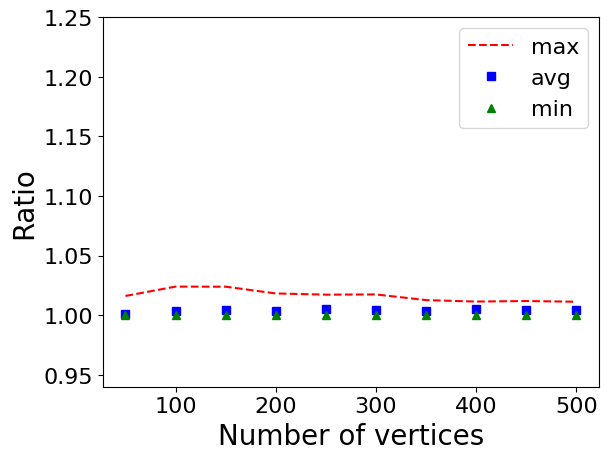}
        \caption{Composite}
    \end{subfigure}
    ~
    \begin{subfigure}[b]{0.31\textwidth}
        \includegraphics[width=\textwidth]{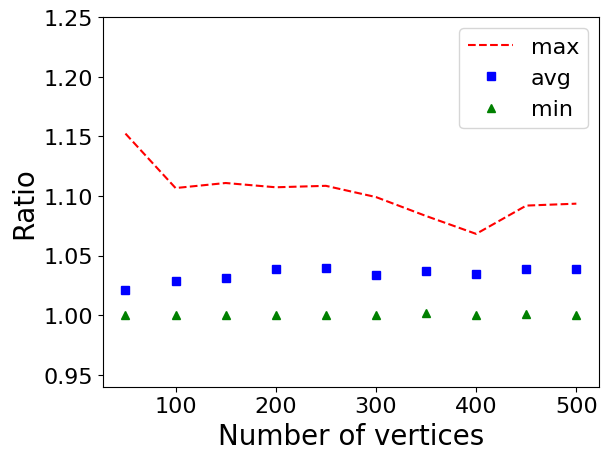}
        \caption{\CMPS}
    \end{subfigure}

    \caption{Performance of \TOP, \CMP, and \CMPS on Watts--Strogatz graphs.
    } 
    \label{WS_N}
\end{minipage}

\vspace{5ex}

\begin{minipage}{\textwidth}
    \centering
    \begin{subfigure}[b]{0.31\textwidth}
        \includegraphics[width=\textwidth]{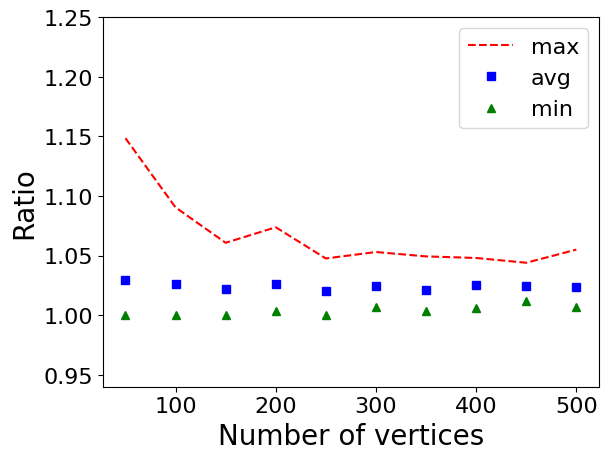}
        \caption{Top-down}
    \end{subfigure}
    ~
    \begin{subfigure}[b]{0.31\textwidth}
        \includegraphics[width=\textwidth]{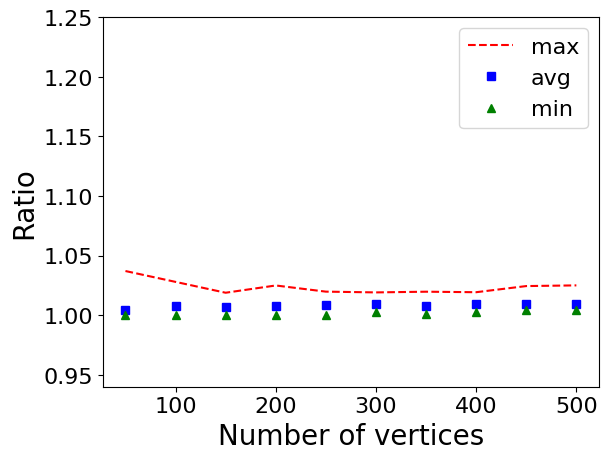}
        \caption{Composite}
    \end{subfigure}
	~
    \begin{subfigure}[b]{0.31\textwidth}
        \includegraphics[width=\textwidth]{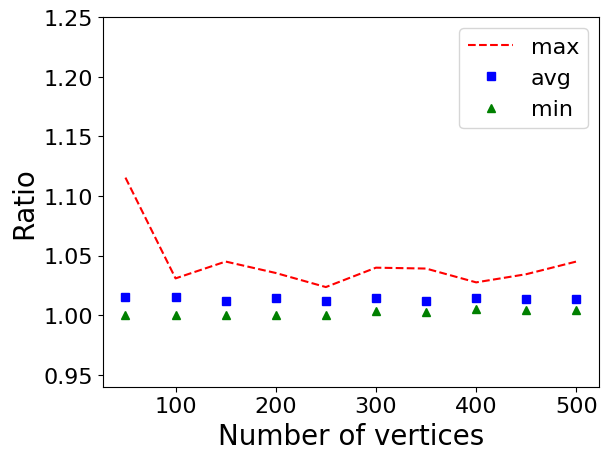}
        \caption{\CMPS}
    \end{subfigure}

\caption{Performance of \TOP, \CMP, and \CMPS on Barab{\'a}si--Albert graphs using the 2-approximation for \STP as a subroutine.} 
    \label{BA_N}
\end{minipage}
\end{figure}

The most time consuming part of this experiment was calculating the exact solutions of all MLST instances. It took 88.64 hours to compute all exact solutions. The computation time for network models ER, WS, and BA were 73.8, 7.84 and 7 hours respectively. Figure ~\ref{BoxPlots_LVT} shows the time taken to compute the exact solutions (with cost \OPT), as a function of the number of levels $\ell$.
As expected, the running time of the heuristics gets worse as $\ell$ increases. Note that the $y$-axes of the graphs in these figures have different scales for different network models. The Erd\H{o}s--R{\'e}nyi network model had the highest running time in the worst case. 

Figure ~\ref{BoxPlots_NDVT} shows the time taken to compute the exact solutions, as a function of the terminal selection method, either \textsc{linear} or \textsc{exponential}. Overall, the running times are slightly worse when the size of the terminal sets decreases exponentially, especially in the worst case.

\begin{figure}
\begin{minipage}{\textwidth}
    \centering
    \begin{subfigure}[b]{0.31\textwidth}
        \includegraphics[width=\textwidth]{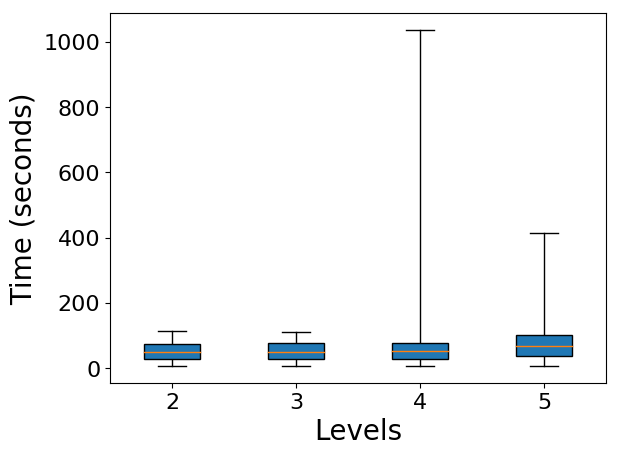}
        \caption{Barab{\'a}si--Albert}
    \end{subfigure}
    ~ 
    \begin{subfigure}[b]{0.31\textwidth}
        \includegraphics[width=\textwidth]{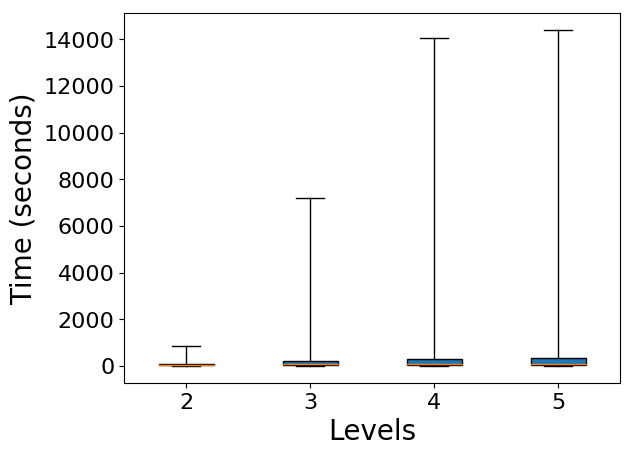}
        \caption{Erd\H{o}s--R{\'e}nyi}
    \end{subfigure}
    ~
    \begin{subfigure}[b]{0.31\textwidth}
        \includegraphics[width=\textwidth]{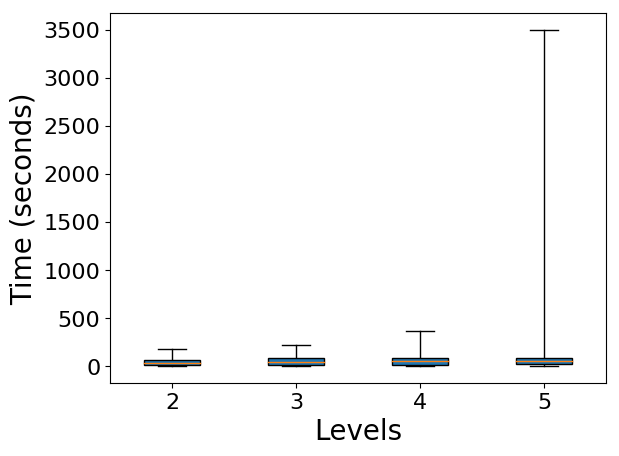}
        \caption{Watts--Strogatz}
    \end{subfigure}
    \caption{Experimental running times for computing exact solutions w.r.t.\ the number~$\ell$ of levels, aggregated over all instances with $\ell$ levels.} 
    \label{BoxPlots_LVT}
\end{minipage}

\vspace{4ex}

\begin{minipage}{\textwidth}
    \centering
    \begin{subfigure}[b]{0.31\textwidth}
        \includegraphics[width=\textwidth]{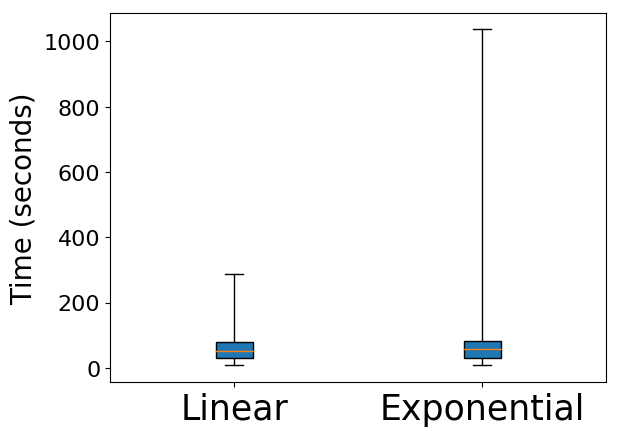}
        \caption{Barab{\'a}si--Albert}
    \end{subfigure}
    ~ 
    \begin{subfigure}[b]{0.31\textwidth}
        \includegraphics[width=\textwidth]{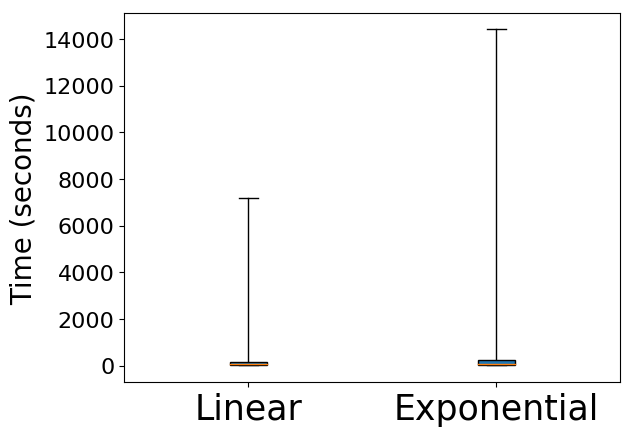}
        \caption{Erd\H{o}s--R{\'e}nyi}
    \end{subfigure}
    ~
    \begin{subfigure}[b]{0.31\textwidth}
        \includegraphics[width=\textwidth]{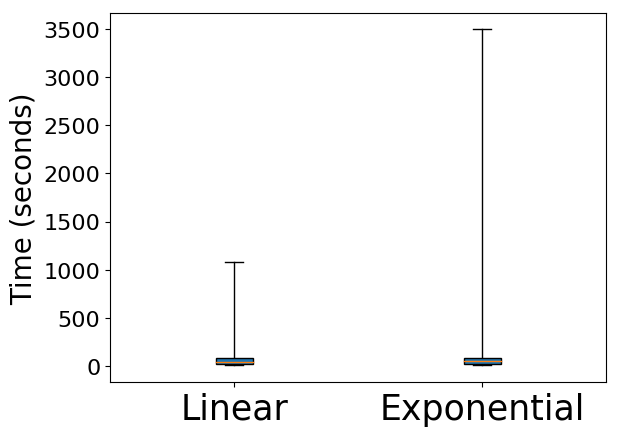}
        \caption{Watts--Strogatz}
    \end{subfigure}
    \caption{Experimental running times for computing exact solutions w.r.t.\ the terminal selection method, aggregated over all instances with \textsc{Linear} or \textsc{Exponential} terminal selection.} 
    \label{BoxPlots_NDVT}
\end{minipage}
\end{figure}

Figure ~\ref{BoxPlots_NVT} shows the time taken to compute the exact solutions, for each of the network models BA, ER, WS, as a function of the number of vertices $|V|$. Since several instances share the same network size, we show minimum, mean, and maximum values. Note that the $y$-axes of the graphs in these figures have different scales for different network models.
As expected, the running time slightly deteriorated as $|V|$ increased, especially in the worst case.

\begin{figure}
    \centering
    \begin{subfigure}[b]{0.31\textwidth}
        \includegraphics[width=\textwidth]{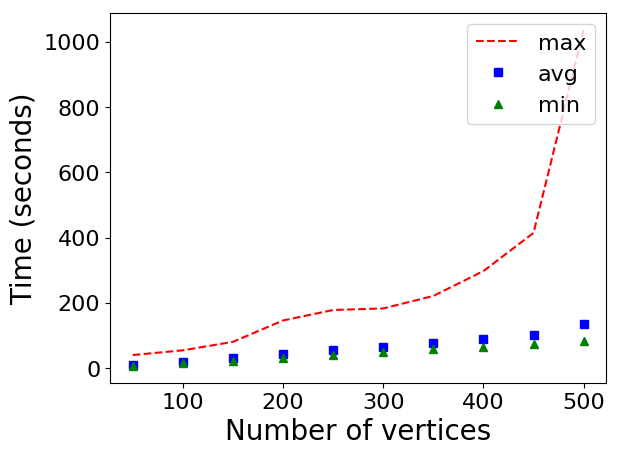}
        \caption{Barab{\'a}si--Albert}
    \end{subfigure}
    ~ 
    \begin{subfigure}[b]{0.31\textwidth}
        \includegraphics[width=\textwidth]{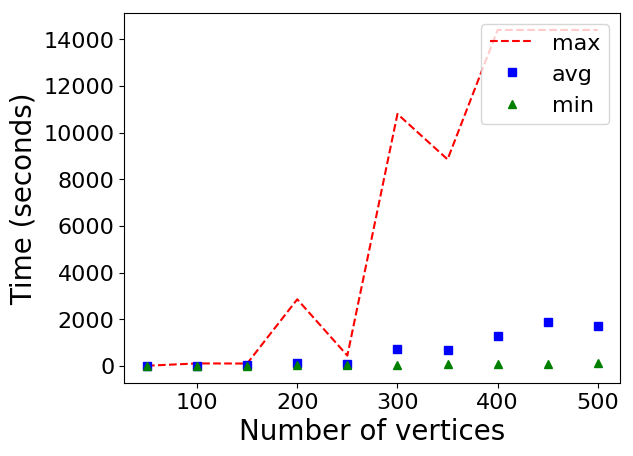}
        \caption{Erd\H{o}s--R{\'e}nyi}
    \end{subfigure}
    ~
    \begin{subfigure}[b]{0.31\textwidth}
        \includegraphics[width=\textwidth]{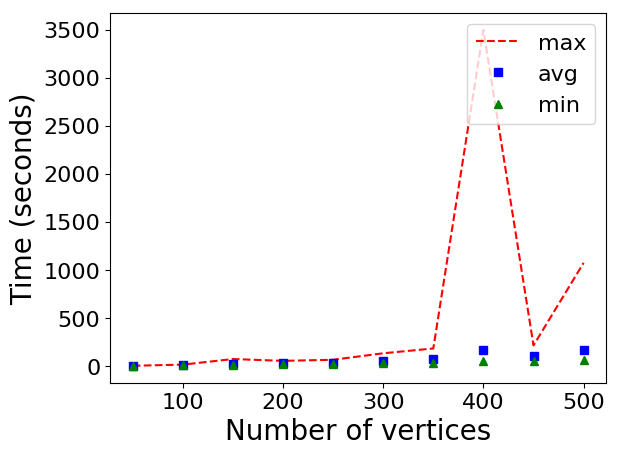}
        \caption{Watts--Strogatz}
    \end{subfigure}
    \caption{Experimental running times for computing exact solutions w.r.t.\ the graph size~$|V|$, aggregated over all instances of $|V|$ vertices.}    \label{BoxPlots_NVT}
\end{figure}%

\section{Conclusions}
\label{section:conclusions}
We presented several heuristics for the MLST problem and analyzed them both theoretically and experimentally. All the software (new heusritcs, approximation algorithms, ILP solvers, experimental data and analysis) are available online \url{ https://github.com/abureyanahmed/multi_level_steiner_ trees}.

The heuristics in this article rely on single level \STP subroutines. The composite heuristic \CMP achieves the best approximation ratio, as it is the minimum of all possible combinations of single level \STP computations. Importantly, we showed that $\CMP(Q^*)$ guarantees the same approximation ratio that \CMP can provide, using at most $2\ell$ \STP computations, rather than $O(2^{\ell-1}\ell)$.
One important question is to consider whether it is possible to directly approximate the MLST problem, without the use of multiple single level \STP subroutines, and whether it is possible to do better than the \CMP approximation ratio. Further, it is natural to study whether there are stronger inapproximability results for the MLST problem, compared to the standard \STP problem.

Another interesting open problem is whether the approximation ratios $t_{\ell}$ (Section~\ref{subsection:composite}) are tight for any $\ell$, and whether the output $\mathbf{y}$ from the LP formulation can help in designing worst-case examples. In particular, even though we have computed the approximation ratio for up to $\ell=100$ levels, it remains to determine the limit $\lim_{\ell\to \infty} t_{\ell}$.

As a final remark, even though our investigation was only focused on the MLST problem, much of the analysis does not depend on the fact that we computed Steiner trees, but only that the computed graphs were nested. We thus wonder whether it is possible to generalize our results to other ``sparsifiers'' (e.g., node-weighted Steiner trees, graph $t$-spanners).

\clearpage

\bibliographystyle{ACM-Reference-Format}
\bibliography{mlst}

\newpage
\section*{Revisions With Respect to Conference Version (SEA 2018)}
We have made the following changes:
\begin{itemize}

\item We added the cut-based ILP formulation (Section \ref{subsection:ilp1}), multi-commodity flow based ILP (section \ref{subsection:ilp2}), as well as a new simplification of the single-commodity ILP (section \ref{subsection:ilp4}) which reduces the number of variables and constraints by a factor of $\ell$. We also added proofs of correctness for the new ILP formulation.

\item The new ILP formulation allowed us to perform experiments with graphs of up to $|V|=500$ vertices. Previously we only tested graphs of up to 100 vertices.

\item All the software (new heusritcs, approximation algorithms, ILP solvers, experimental data and analysis) are available online \url{ https://github.com/abureyanahmed/multi_level_steiner_ trees}.

\item More thoroughly explained the composite heuristic (Section \ref{subsection:composite}), added Lemma \ref{lemma:cmp}, and provided a better treatment of Theorem \ref{theorem:composite-lp} with additional discussion.

\item Derived how to compute an approximation ratio for $\text{CMP}(\mathcal{Q})$ for any $\mathcal{Q} \subseteq \{1,2,\ldots,\ell\}$ in section \ref{subsection:composite}. This yields an alternate proof of the $\frac{\ell+1}{2}$, $\ell$, and 4-approximations for top-down, bottom-up, and the QoS algorithms.

\item Found an approximation ratio for $\CMP$ using column-generation techniques that allows us to compute the exact approximation ratio in Figure \ref{fig:lvr} up to $\ell=100$, which is much larger than any practical need. Previously, we were only able to compute values up to $\ell=22$ levels.

\item Included experimental running times for computing OPT in Figures \ref{BoxPlots_LVT}, \ref{BoxPlots_NDVT}, \ref{BoxPlots_NVT}. These show min, average, and max running times depending on number of levels $\ell$, terminal selection method, and number of vertices $|V|$.

\item Experimented with the use of an exact Steiner tree subroutine and the use of the 2-approximation algorithm in the proposed algorithms. 
We found that the four algorithms perform very well in practice with the 2-approximation algorithm as a (single-level) \STP subroutine, and that using an exact \STP subroutine does not noticeably improve the performance. Hence, we only discuss the results that use the 2-approximation algorithm.

\item We expanded and improved the Conclusions section, discussing several open problems.

\item We changed the notation for the number of levels from $k$ to $\ell$ and renumber them 1, \ldots, $\ell$ with level $\ell$ on top (previously $k$, \ldots, 1 with 1 on top). This is more intuitive as it clearly indicates that an edge on level $i$ (but not $i+1$) pays $i$ times its cost (rather than $k+1-i$ times as before).

\item Clarified and improved several proofs, most notably the TD and BU algorithms, and added new figures.

\end{itemize}

\end{document}